\documentclass[aps,prx,a4paper,onecolumn,superscriptaddress,longbibliography,10pt]{revtex4-2}
\usepackage{graphicx,bm,amssymb,mathtools,mathrsfs,amsthm,bbm,appendix,amsmath,amsfonts,physics,algorithm2e,enumitem, multirow} 

\usepackage[table]{xcolor}
\definecolor{ncolor4}{HTML}{C70694}
\definecolor{ncolor3}{rgb}{0.12, 0.3, 0.17}
\definecolor{ncolor2}{rgb}{0.77, 0.12, 0.23}
\definecolor{ncolor5}{rgb}{0.36, 0.22, 0.33}
\definecolor{ncolor1}{rgb}{0.64, 0.0, 0.0}
\definecolor{ncolor6}{rgb}{0.03, 0.27, 0.49}
\definecolor{myblue}{HTML}{CFCBFF}
\usepackage[caption=false]{subfig}
\usepackage[unicode=true,pdfusetitle, bookmarks=true,bookmarksnumbered=false,bookmarksopen=false, breaklinks=false,pdfborder={0 0 0},backref=false,colorlinks=true, linkcolor=blue,citecolor=blue,urlcolor=blue]{hyperref}

\newcommand{\erg}[1]{\mathrm{Erg}\left(#1\right)}
\newcommand{\ergobs}[1]{\mathrm{Erg}_{\mathrm{obs}}\left(#1\right)}
\newcommand{\wt}[1]{\widetilde{#1}}
\newcommand{\overbar}[1]{\mkern 2mu\overline{\mkern-2mu#1\mkern-2mu}\mkern 2mu}
\newcommand{\mc}[1]{\mathcal{#1}}

\newcommand{\ug}{U_{\mathrm{gl}}}
\newcommand{\vg}{V_{\mathrm{gl}}}
\newcommand{\ugdag}{U_{\mathrm{gl}}^{\dag}}
\newcommand{\vgdag}{V_{\mathrm{gl}}^{\dag}}
\def\Pr{\mathrm{Pr}}  
\def\Tr{\mathrm{Tr}}

\newenvironment{protocol}[1][htb]
  {
   \begin{algorithm}[#1]%
  }{\end{algorithm}}

\DeclareOldFontCommand{\rm}{\normalfont\rmfamily}{\mathrm}

\newtheorem{theorem}{Theorem}

\newtheorem{lemma}{Lemma}

\newtheorem{remark}{Remark}

\begin{document}
\title{Sample Complexity of Black Box Work Extraction} 

\author{Shantanav Chakraborty}
\email{shchakra@iiit.ac.in}
\affiliation{Centre for Quantum Science and Technology (CQST), International Institute of Information Technology Hyderabad, Gachibowli 500032, Telangana, India}
\affiliation{Center for Security, Theory and Algorithmic Research (CSTAR), International Institute of Information Technology Hyderabad, Gachibowli 500032, Telangana, India}

\author{Siddhartha Das}
\email{das.seed@iiit.ac.in}
\affiliation{Centre for Quantum Science and Technology (CQST), International Institute of Information Technology Hyderabad, Gachibowli 500032, Telangana, India}
\affiliation{Center for Security, Theory and Algorithmic Research (CSTAR), International Institute of Information Technology Hyderabad, Gachibowli 500032, Telangana, India}

\author{Arnab Ghorui}
\email{arnab.ghorui@research.iiit.ac.in}
\affiliation{Centre for Quantum Science and Technology (CQST), International Institute of Information Technology Hyderabad, Gachibowli 500032, Telangana, India}
\affiliation{Center for Security, Theory and Algorithmic Research (CSTAR), International Institute of Information Technology Hyderabad, Gachibowli 500032, Telangana, India}

\author{Soumyabrata Hazra}
\email{soumyabrata.hazra@research.iiit.ac.in}
\affiliation{Centre for Quantum Science and Technology (CQST), International Institute of Information Technology Hyderabad, Gachibowli 500032, Telangana, India}
\affiliation{Center for Computational Natural Sciences and Bioinformatics (CCNSB), International Institute of Information Technology Hyderabad, Gachibowli 500032, Telangana, India}

\author{Uttam Singh}
\email{uttam@iiit.ac.in}
\affiliation{Centre for Quantum Science and Technology (CQST), International Institute of Information Technology Hyderabad, Gachibowli 500032, Telangana, India}
\affiliation{Center for Security, Theory and Algorithmic Research (CSTAR), International Institute of Information Technology Hyderabad, Gachibowli 500032, Telangana, India}

\begin{abstract}
Extracting work from a physical system is one of the cornerstones of quantum thermodynamics. The extractable work, as quantified by ergotropy, necessitates a complete description of the quantum system. This is significantly more challenging when the state of the underlying system is unknown, as quantum tomography is extremely inefficient. In this article, we analyze the number of samples of the unknown state required to extract work. With only a single copy of an unknown state, we prove that ergotropy approaches zero in the asymptotic limit, rendering work extraction nearly impossible. In contrast, when multiple copies are available, we quantify the sample complexity required to estimate extractable work, establishing a scaling relationship that balances the desired accuracy with success probability. Our work develops a sample-efficient protocol to assess the utility of unknown states as quantum batteries and opens avenues for estimating thermodynamic quantities using near-term quantum computers.
\end{abstract}
\date{\today}
\maketitle

\section{ Introduction}
\label{sec:intro}
Quantum thermodynamics \cite{Goold2016, Vinjanampathy2016, Thermo2018, Deffner2019, Bera2019, Potts2024} is vital for developing and optimizing a range of quantum technologies, including computation \cite{Auffeves2022, Cleri2024,  Arora2024}, heat engines \cite{Alicki1979, Scully2002, Dechant2014, Rossnage2016, Maslennikov2019, Klatzow2019}, batteries \cite{Ferraro2018, Yang2023, Campaioli2024}, and sensors \cite{Chu2022, Kominis2024}. 
A central focus of this field is the study of work extraction from quantum systems \cite{Scully2003, Michal2013, Brandao2015b, Skrzypczyk2014, Singh2019, Pinto2024}, which has significant implications for the development and advancement of quantum technologies.~A better understanding of the amount of extractable work from a given system is crucial for characterizing quantum batteries, powering quantum technologies, and, more broadly, building thermodynamically efficient and scalable quantum devices \cite{Acin2018, Buffoni2022, Myers2022, Fellous2023, Maillette2023, Moutinho2023, Arrachea2023, Jaschke2023, Van2024, Majidi2024, Patryk2024}.

For an $N$-dimensional quantum system in a state $\rho$ with Hamiltonian $H$, the extractable work under unitary transformations is quantified by the ergotropy of the state $\rho$, given by \cite{Lenard78, Pusz1978, Allahverdyan2004}
\begin{align}
\label{eq:erg-def-first}
    \erg{\rho}&=\max_{U\in \mathsf{U}(N)}\Tr\left[H\left(\rho-U\rho U^{\dag}\right)\right],
\end{align}
where $\mathsf{U}(N)$ is the group of $N\times N$ unitary matrices. In recent years, ergotropy has been extensively explored across various information processing and computing tasks \cite{Farina2019, Singh2021, Sen2021, swati2023, simon2024, Castellano2024, Bhattacharyya2024, joshi2024, xuereb2024resources, Castellano2024b}.  

To estimate the maximal amount of extractable work by measuring ergotropy, it is necessary to optimize over unitaries $U$, as shown in Eq.~\eqref{eq:erg-def-first}. The resulting optimal unitary depends on the spectrum of $\rho$, which is revealed only when a complete description of the state of a quantum system is available. However, in realistic settings, due to technological limitations, the presence of uncontrollable environmental noise, and/or untrusted sources, the state in which the system is prepared is often unknown or, at best, only partially known \cite{Caves2002, Alvarez2015, Szangolies2015, Oszmaniec2016, Li2019, Lebedev2020, cheng2024, pensia2024, zhu2024, watanabe2024, Zambrano2024}. This makes it impossible to estimate ergotropy (and hence, the maximal possible extractable work) without performing a full tomography on $\rho$, a highly inefficient procedure \cite{Cramer2010, Odonnell2015}. The best we can hope for is to devise a strategy in such a black-box setting that still provides some information about the amount of extractable work while requiring only a few samples of $\rho$. This naturally leads us to the following question: In such a black box setting, where the underlying quantum state is unknown, how many samples are needed to obtain an estimate of the amount of extractable work from the system? Despite the importance of a sample-efficient scheme for work extraction, this question has remained largely unexplored. In this paper, we address this gap by examining how the number of copies of $\rho$ affects work extraction.
\begin{figure}[ht!]
\centering
\includegraphics[width=0.45\textwidth]{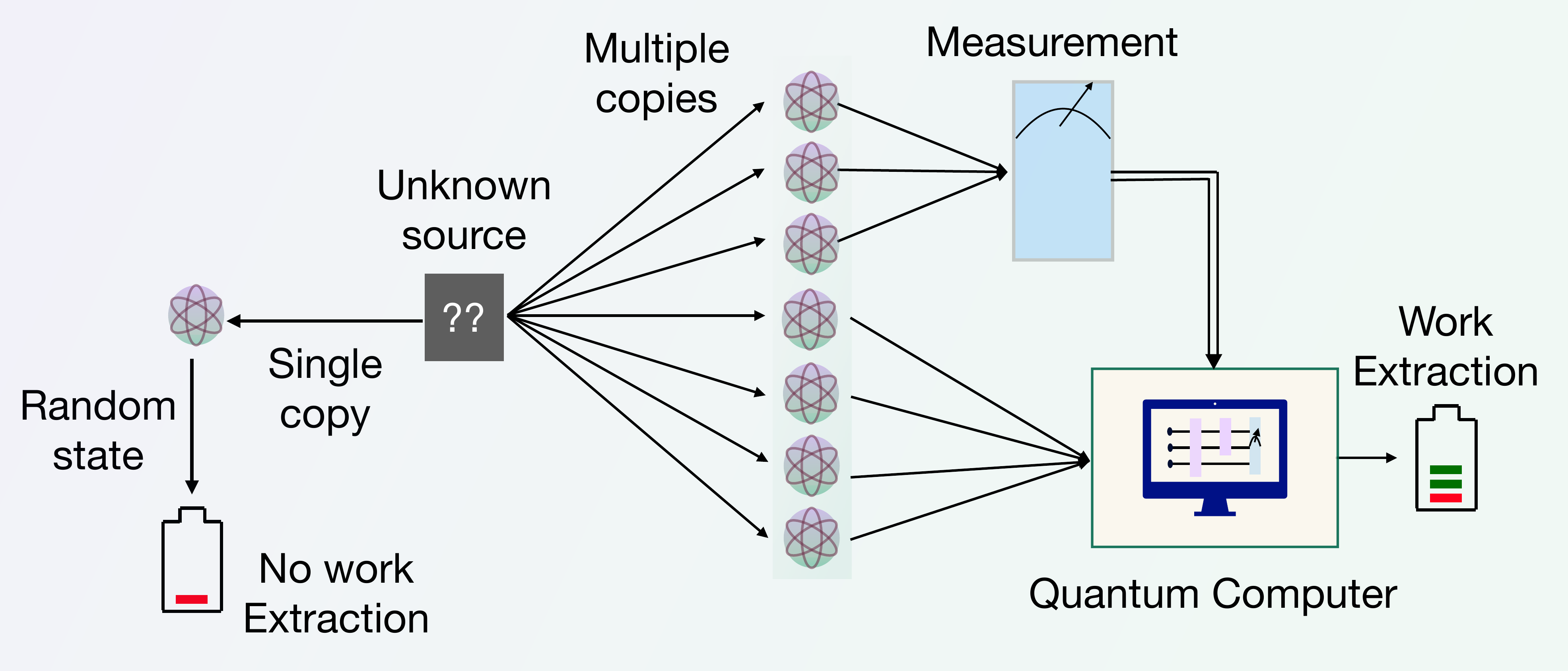}
\caption{Schematic of the possibilities for extracting work when a quantum system is prepared in some unknown state $\rho$. We show that it is nearly impossible to extract any work with only a single copy of $\rho$. On the other hand, when multiple copies are available, we provide a sample and qubit-efficient randomized quantum algorithm for work extraction that requires only a few samples of $\rho$.}
\label{fg.schem}
\end{figure}

One possibility is to proceed by assuming that the $N$-dimensional unknown state $\rho$ is random. That is, the purification of $\rho$ leads to a Haar random state of dimension $NN'$ (for large environment dimension $N'$). Our first contribution is to rigorously establish that it is nearly impossible to extract any (non-zero) work from a single copy of the unknown random quantum state in the asymptotic limit. This is obtained by proving that ergotropy is Lipschitz continuous with the Lipschitz constant scaling as the spectral norm of the Hamiltonian of the system, allowing us to use the concentration of measure phenomenon (such as Levy's lemma \cite{Ledoux2005}) to show that for any random $\rho$, $\mathrm{Erg}(\rho)$ is close to its expectation value, which approaches zero as $N'\rightarrow\infty$. Owing to its connections with quantifying work, there is widespread interest in estimating ergotropy and understanding its properties. Consequently, our result is of independent interest. On the other hand, as mentioned before, tomography would require exponentially many samples and measurements to estimate ergotropy. This points to a trade-off between the number of available copies of the unknown state and extractable work. Interestingly, in Ref.~\cite{vsafranek2023work}, the authors propose a scheme using partial information from quantum measurements that quantifies the extractable work from unknown quantum states by associating it with a quantity called observational ergotropy. However, it is unclear if this quantity can be estimated faithfully in a sample-efficient manner.

Our second contribution equips this procedure with rigorous performance guarantees.~More precisely, we develop a sample-efficient randomized protocol to estimate observational ergotropy, facilitating work extraction. Our method computes this quantity to $\varepsilon$-additive accuracy, with high probability, by making very few measurements on $\rho$ -- this requires $O(1/\varepsilon^2)$ copies of $\rho$ and, remarkably, has only a logarithmic dependence on the dimension of the system. The protocol is a randomized quantum algorithm that can be efficiently run on a quantum computer using as many samples of $\rho$. Finally, we establish rigorous robustness bounds, showing that observational ergotropy can still be reliably estimated under realistic errors. Our result thus broadens the applicability of quantum computers to the efficient estimation of extractable work and the characterization of quantum batteries. 

The article is organized as follows. In Sec.~\ref{sec:main-results}, we give an overview of our main findings, namely the two distinct approaches to work extraction. Sec.~\ref{subsec:no-work-single-copy} shows that the ergotropy of a single copy of a random state $\rho$ approaches zero in the asymptotic limit. In contrast, when multiple copies of $\rho$ are available, Sec.~\ref{subsec:work-multiple-copy} presents a sample-efficient randomized quantum algorithm for estimating observational ergotropy. The robustness of this estimation is analyzed in Sec.~\ref{sec:robustness-ergobs}. Preliminaries are introduced in Sec.~\ref{prelim}, which lay the foundation for the technical proofs in Sec.~\ref{sec:methods}. Detailed derivations of the results from Secs.~\ref{subsec:no-work-single-copy}, \ref{subsec:work-multiple-copy}, and \ref{sec:robustness-ergobs} are provided in Secs.~\ref{append-average-ergotropy}, \ref{subsec:proof-multiple-copies}, and \ref{sub:proofs-lems}, respectively. In Sec.~\ref{subsec:pauli-measurement}, we also outline an alternative quantum algorithm that efficiently estimates extractable work when the system Hamiltonian is local. Sec.~\ref{sec:discussion} concludes the paper. Finally, the Appendix presents several complementary results and illustrative examples that further reinforce the theoretical findings.

\section{Main Results}
\label{sec:main-results}
In this section, we outline the main findings of our work, which are summarized in Fig.~\ref{fg.schem}. Throughout this work, we consider $n$-qubit quantum systems associated with a Hilbert space $\mathcal{H}_N$ of dimension $N=2^n$, and the underlying system is described by a Hamiltonian $H$, which we have full knowledge of. Suppose $\mathcal{D}_N$ denotes the set of density operators on $\mathcal{H}_N$, then, in our setting, the $n$-qubit system of Hamiltonian $H$ is prepared in some unknown quantum state $\rho\in\mathcal{D}_N$.

\subsection{(No) work extraction from a single copy of $\rho$} 
\label{subsec:no-work-single-copy}
When only a single copy of the unknown state is available, it is, in principle, possible to consider $\rho$ to be a random state drawn from some statistical distribution over all possible density matrices. In this picture, we assume that a single query to a black box produces a state sampled from a known distribution, motivated by the fact that with only one copy of $\rho$ there is little else that the experimentalist can exploit. Thus, in this black box setting, although the underlying distribution is known, the specific instance of the state, produced by the black box, remains unknown. The idea then is to use information about this distribution to extract some work. To this end, we assume that $\rho \in \mc{D}_N$ is a random full rank density matrix whose purification leads to a Haar random state of dimension $NN'$. We prove that it is almost always impossible to extract any nontrivial amount of work (ergotropy) with high probability over possible density matrices from this distribution. This result is obtained via a sequence of non-trivial steps, which we outline here, while detailed derivations can be found in Sec.~\ref{append-average-ergotropy}.

Since $\rho$ is a random quantum state, its ergotropy, $\mathrm{Erg}(\rho)$, is a random variable. For any Hamiltonian $H$ of bounded spectral norm $\|H\|$, we first prove that this expectation value is upper bounded by a quantity which approaches zero as $N'\rightarrow\infty$. More precisely, (See  Lemma \ref{lem:avg-ergotropy} of Sec.~\ref{append-average-ergotropy}), we prove that
$$
\mathbb{E}\left[\mathrm{Erg}(\rho)\right]\leq \|H\|\sqrt{\dfrac{N}{N'}}.
$$
We then estimate for a random instance of $\rho$ how close $\mathrm{Erg}(\rho)$ is to its expected value. Typically, this is proved using measure concentration bounds, which require Lipschitz continuity. However, whether the latter property holds for the ergotropy function has been unknown until recently \cite{hovhannisyan2024concentration}. We develop a slightly different approach to prove that ergotropy is a Lipschitz continuous function with a Lipschitz constant of $4\norm{H}$. This allows us to apply the concentration of measure bound to derive that with a high probability $\mathrm{Erg}(\rho)\leq \gamma$, where $\gamma$ is a small number approaching zero. Our results are formally stated via the following theorem.
\begin{figure}[ht!]
\centering
\subfloat[N=2]{
{\includegraphics[width=0.4\textwidth]{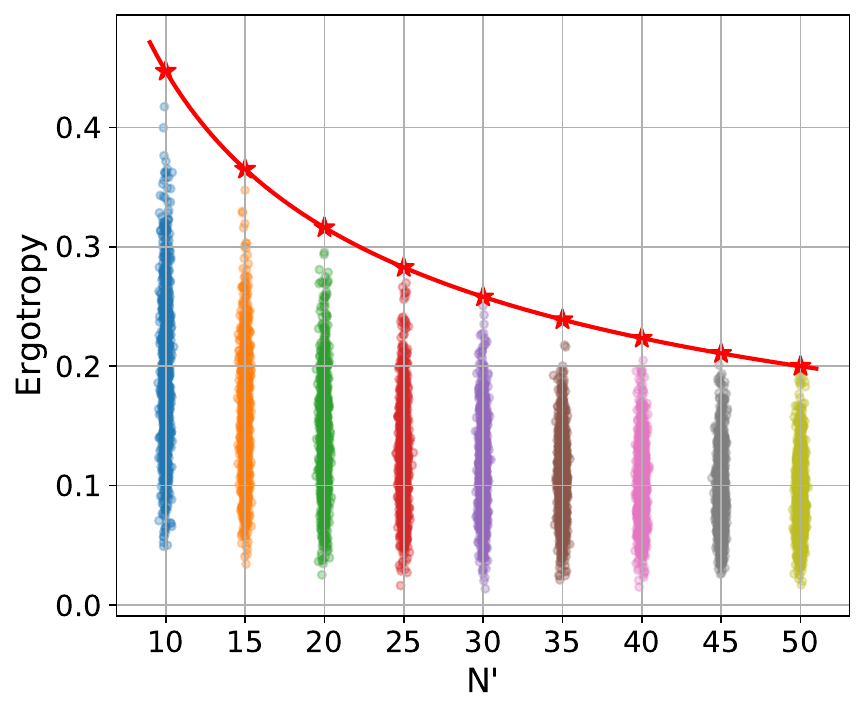}\label{fg.N2}}
}
\subfloat[N=3]{
{\includegraphics[width=0.4\textwidth]{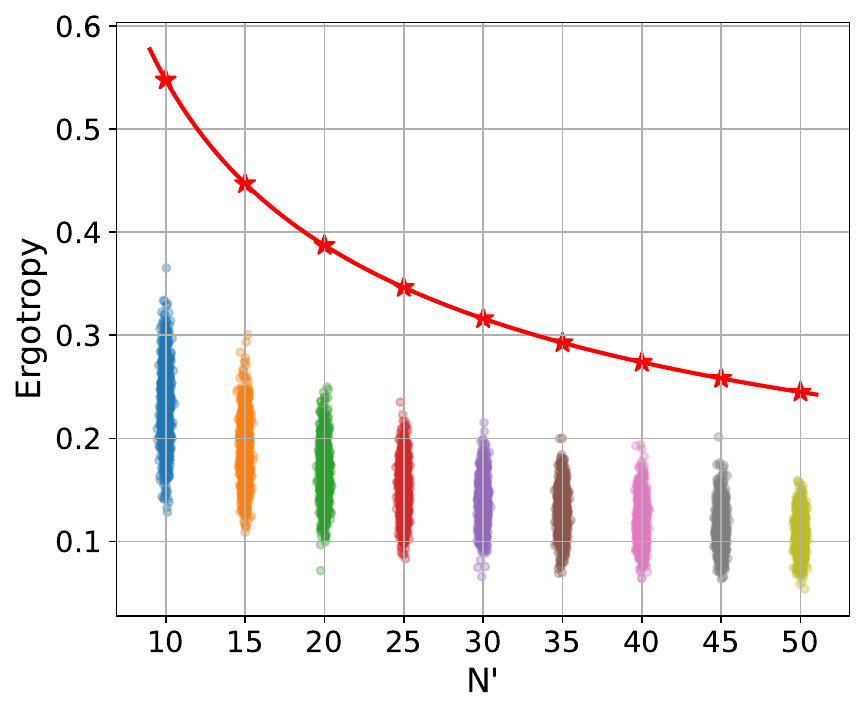}\label{fg.N3}}
}\\
\subfloat[N=4]{
{\includegraphics[width=0.4\textwidth]{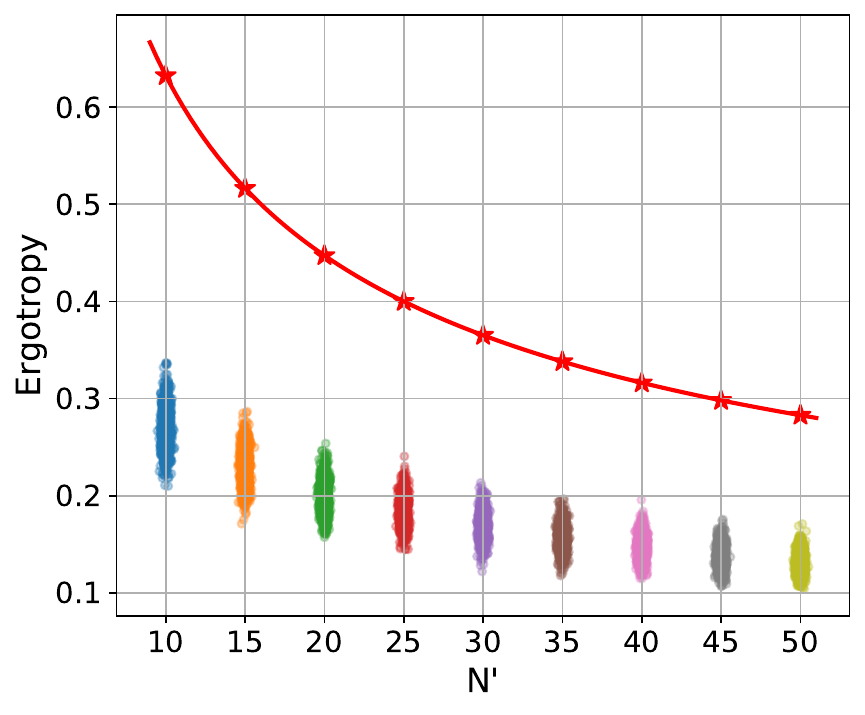}\label{fg.N4}}
}
\subfloat[N=5]{
{\includegraphics[width=0.4\textwidth]{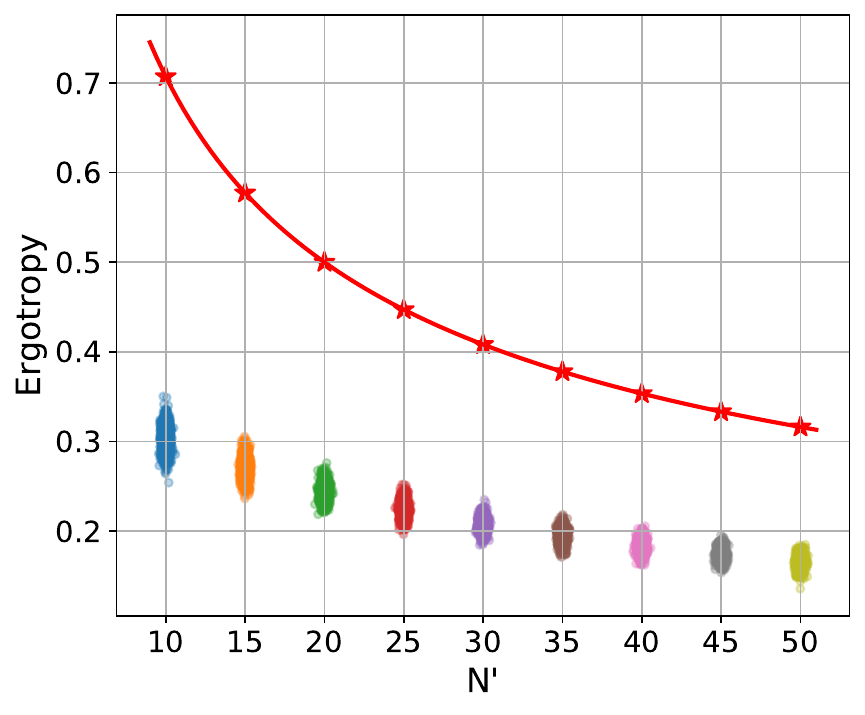}\label{fg.N5}}
}
\caption{Ergotropy of random states. We plot the ergotropy of random quantum states (generated by partial tracing $NN'$-dimensional Haar random pure states) for $N=2$ (Fig.~\ref{fg.N2}), $N=3$ (Fig.~\ref{fg.N3}), $N=4$ (Fig.~\ref{fg.N4}), and $N=5$~(Fig.~\ref{fg.N5}), as a function of $N'$. The Hamiltonian is considered to be $H=\sum_{i=1}^N h_i/N$, where $h_i=\ket{1}\bra{1}$, and $\norm{H}=1$. For each $N$ and $N'$, we sample several random density matrices and compute the ergotropy for each of them (indicated by the cluster of points for each $N, N'$). The red curve in each figure is a plot of the function $\sqrt{N/N'}$, which is an upper bound on ergotropy of random states with high probability as established in Theorem \ref{thm:impossibility-work}. We find that our upper bound becomes increasingly tighter as the ratio $N/N'$ decreases.}
\label{fg.erg-conc}
\end{figure}
\begin{theorem}[Impossibility of work extraction from a random state]
\label{thm:impossibility-work}
    Let $\rho(U) \in \mc{D}_N$ be a full rank random density matrix obtained by partial tracing subsystem $B$ of a randomly selected pure state $U\ket{\psi_0}_{AB}$ in a composite Hilbert space of dimension of $NN'$, where $\ket{\psi_0}_{AB}$ is some fixed state and $N'>N$. Then for $\gamma >2\norm{H}\,\sqrt{N/N'}$, the ergotropy of $\rho(U)$ satisfies
\begin{align}
\label{eq:main-th-main-text}
    \Pr\left[ \mathrm{Erg}(\rho(U))  > \gamma \right] \leq \exp\left[-\frac{NN'\gamma^2}{768 \,\norm{H}^2}\right].
\end{align}
\end{theorem}
We prove this theorem in detail in Sec.~\ref{subsec:no-work-single-copy}. Note that the tolerance $\gamma$ of ergotropy goes to zero in the limit $N'\rightarrow\infty$. Also, the right-hand side of Eq.~\eqref{eq:main-th-main-text} vanishes in the limit $N'\rightarrow\infty$ for any $\gamma$, however small. Thus Theorem \ref{thm:impossibility-work} can be interpreted as follows: it is impossible (assuming that the Hamiltonian is bounded) to extract any ergotropy from a full rank random quantum state in $\mc{D}_N$ obtained by partial tracing randomly selected pure state in a composite Hilbert space of dimension of $NN'$, in the limit $N'\rightarrow \infty$. We numerically calculate the ergotropy of random quantum states (obtained as mentioned previously) with increasing $N'$, for a fixed Hamiltonian $H$. The results, shown in Fig.\ref{fg.erg-conc}, indicate that the upper bound on $\mathrm{Erg}(\rho)$ from Theorem \ref{thm:impossibility-work}, becomes increasingly tight as the ratio $N/N'$ decreases.

One expects that this limitation can be overcome given access to multiple copies of the unknown state $\rho$. In the next section, we consider a different black box setting and rigorously quantify the number of samples of $\rho$ required to quantify the amount of work that can be extracted from it. 

\subsection{Extracting work from multiple copies of $\rho$}
\label{subsec:work-multiple-copy}
In this section, we consider a different setting: each query to a black box returns the same fixed but unknown state $\rho$. By allowing multiple queries, one can obtain multiple copies of this state. Crucially, unlike the single-copy scenario discussed earlier, this multi-copy black-box model does not involve sampling from a random ensemble: the state $\rho$ is fixed (though unknown) throughout. We show that, in this setting, the observational ergotropy, which quantifies the extractable work, can be estimated in a sample-efficient manner using a quantum computer. 

As outlined in Ref.~\cite{vsafranek2023work}, this quantity can be obtained via the following two-stage process:~(i) Firstly, perform projective measurement with mutually orthogonal operators $\{ {P}_i\}_{i=1}^M$ such that $\sum_{i}P_i=\mathbb{I}_N$, and compute $p_i=\Tr[P_i \rho]$, for each $i$. This step essentially decomposes the full space into $M$ blocks of orthogonal support, i.e.\ $\oplus_{i=1}^{M}\mathcal{H}_i$, where $\mathcal{H}_i=P_i \mathcal{H}_N P_i
$. We assume that the dimension of each block is $2^r$, where $r= \log_2(N/M)$. (ii) Secondly, randomize each block completely by applying a direct sum of Haar random unitaries to convert the still unknown state from the previous step into a known state $\widetilde{\rho}$. That is, \begin{align}
\label{eq:randomized-state}
\wt{\rho}&=\int_{\wt{V}_i\sim U(2^r)} d\mu_{\mathrm{Haar}}\left(\oplus_{i=1}^{M} \wt{V}_i\right)\times \left[\left(\oplus_{i=1}^{M} \wt{V}_i\right) \rho \left(\oplus_{i=1}^{M} \wt{V}^{\dag
}_i\right) \right],
\end{align}
where $d\mu_{\mathrm{Haar}}$ is the normalized Haar measure, which factorizes nicely as $d\mu_{\mathrm{Haar}}\left(\oplus_{i=1}^{M} \wt{V}_i\right)=\prod_{i} d\mu_{\mathrm{Haar}}\left(\wt{V_i}\right)$. These steps allow us to compute the observational ergotropy of the original unknown state $\rho$, which is defined as (see Sec.~\ref{prelim} and Refs.~\cite{Safranek2019, Safranek2021, vsafranek2023work})
\begin{align}
\label{eq:reg-obs-first}
\ergobs{\rho}=\Tr\left[H~\left(\rho-U_{\mathrm{gl}}~\wt{\rho}~ U^{\dag}_{\mathrm{gl}}\right)\right],
\end{align}
where $U_{\mathrm{gl}}$ is a global unitary implementing a basis change from the basis of space $\oplus_{i=1}^{M}\mathcal{H}_i$ to the eigenbasis of $H$ and a permutation that orders the $p_i$'s in non-increasing order. This has the effect of constructing the passive state with respect to $\wt{\rho}$, analogous to the role of the optimal unitary in the expression for ergotropy. As there is complete knowledge of $H$, all that is required to construct $U_{\mathrm{gl}}$ are the $p_i$'s obtained from the first stage. Since $\widetilde{\rho}$ is obtained by extracting partial information from the original state, the quantity $\mathrm{Erg}_{\mathrm{obs}}(\rho)$ provides a lower bound on $\mathrm{Erg}(\rho)$, with the two quantities matching when $\rho$ is completely known. In fact, in Appendix \ref{append-deviation-erg-ergobs}, we rigorously derive a non-trivial lower bound on the deviation between these two quantities.

We first quantify the number of samples of $\rho$ required to obtain an estimate $\wt{p_i}$ of~$p_i=\Tr[P_i\rho]$, for all $i\in [1,M]$. Performing a projective measurement results in outcome $i$ with probability $p_i$. The strategy is to make $\wt{T}$ repeated projective measurements, each time with a fresh copy of $\rho$: for the $j^{\mathrm{th}}$ iteration, the measurement outcome $i$ corresponds to an i.i.d. random variable $Y^{(j)}_{i}$ whose expectation value is $p_i$. So, for $i\in [1, M]$, we take the sample mean of each of the $\wt{T}$ runs, i.e. $\wt{p_i}=\sum_{j=1}^{\wt{T}} Y^{(j)}_{i}/\wt{T}$. Then, in Sec.~\ref{subsec:work-multiple-copy} we prove that Hoeffding's inequality \cite{Hoeffding1963} ensures that $|\wt{p_i}-p_i|\leq \varepsilon$, with a failure probability of at most $\delta/M$, as long as $\wt{T}$ is large enough. That is, if $G_i$ is the event that $|\wt{p_i}-p_i|>\varepsilon$, then $\Pr[G_i]\leq \delta/M$. Finally, using the union-bound, we derive that only $\wt{T}$ samples are sufficient to ensure that for every $i\in [1, M]$, $|\wt{p_i}-p_i|\leq \varepsilon$, with a probability of at least $1-\delta$. 

Observe that the construction of the global unitary $U_{\mathrm{gl}}$ necessitates arranging the set of estimates $\{\wt{p_i}\}_{i=1}^{M}$ in non-increasing order. That is, the accuracy with which these estimates approximate each $p_i$ should be such that they leave this ordering unaltered, i.e., \ $p_i-p_j\geq 0$ should also imply $\wt{p_i}-\wt{p_j}\geq 0$, for all pairs of $ i, j\in [1, M]$. Thus, if 
$$
\Delta := \min_{i,j}\{|p_i-p_j|: p_i\neq p_j\},
$$ 
this is ensured by choosing any $\varepsilon < \Delta/2$. Despite requiring only a few measurements, this estimation scheme provides sufficient information regarding $\rho$ to enable the construction of $U_{\mathrm{gl}}$. Formally, we prove the following lemma in Sec.~\ref{subsec:proof-multiple-copies}:

\begin{lemma}
\label{lem:sample-prob-with-error}
 The number $\wt{T}$ of samples of $\rho$ required to estimate $M$ probabilities upto an error $\Delta/3$, i.e., $\left| \wt{p}_i-p_i\right| \leq \Delta/3$ for all $i\in [M]$ with failure probability at most $\delta$, is
\begin{align}
\label{eq:samp-prob-prot}
    \wt{T} =O\left(\dfrac{\log (M/\delta)}{\Delta^2}\right).
\end{align}
\end{lemma}

Having estimated every $\wt{p}_i$ with sufficient accuracy in a sample-efficient manner, we are now in a position to estimate observational ergotropy for unknown state $\rho$. The central idea is to develop an easy-to-implement, efficient, randomized scheme that iteratively estimates the desired quantity. This is given by Protocol \ref{prot1}. Indeed, the outcome of this randomized protocol is a number that approximates the observational ergotropy to $\varepsilon$-additive accuracy, with a success probability of at least $1-\delta$. We outline the details of this scheme.
\RestyleAlgo{boxruled}
\begin{protocol}
\caption{$\texttt{Work Extraction from unknown~}\rho$}
  \label{prot1}
\begin{enumerate}
    \item With probability $1/2$, measure $H$ on the state $\rho$.
    \item With probability $1/2$, execute steps $(a)-(c)$.
    \begin{enumerate}
        \item Measure $\rho$ in the basis spanned by the projectors $\{P_i\}_{i=1}^{M}$ and ignore the outcomes to obtain the post-measurement state $\bar{\rho}$.
        \item Sample a unitary $U_r$ uniformly at random from the $r$-qubit Pauli group, where $r=\log_2(N/M)$, and apply the unitary $\ug(\mathbb{I}_M\otimes U_r)$ to $\bar{\rho}$ to obtain
        \begin{align*}
        \rho'=\ug(\mathbb{I}_M\otimes U_r) \bar{\rho} (\mathbb{I}_M\otimes U_r)^{\dag}\ugdag
        \end{align*}
        \item Measure $-H$ in the state obtained after Step 2(b).
    \end{enumerate}   
  \item Store the outcome of the measurement of the $j^{\mathrm{th}}$ iteration into $\mu_j$.
  \item Repeat Steps 1 and 2, a total of $T$ times. 
  \item Output $\mu=\dfrac{2}{T}\sum_{j=1}^{T} \mu_j$.
\end{enumerate}
\end{protocol}

As the outcome of the measurement on $\rho$ is ignored, Step~2(a) naturally coarse-grains the state, on average. Note that each $P_i$ acts on the whole Hilbert space $\mathcal{H}_N$, but has support only on $\mathcal{H}_i$. If $\rho_{ii}$ denotes the projection of $\rho$ on the support of $P_i$, the post-measurement state $\bar{\rho}$, averaged over all outcomes, is
$$
\mathbb{E}[\bar{\rho}]=\sum_{i=1}^{M}  P_i\rho P_i=\bigoplus_{i=1}^{M} \rho_{ii}.
$$ 
This is a block-diagonal state, where each block corresponds to the restriction of $\rho$ to $\mathcal{H}_i$. To obtain the desired coarse-grained state $\wt{\rho}$ (as defined in Eq.~\eqref{eq:randomized-state}), it suffices to apply a direct sum of Haar-random unitaries to $\mathbb{E}[\bar{\rho}]$. However, doing so efficiently is essential for preserving the overall sample complexity of the protocol. 

To achieve this, Step~2(b) samples a single unitary $U_r$ uniformly at random from the $r$-qubit Pauli group $\mc{P}(r)$, which forms an exact unitary 1-design consisting of $L = 4^{r+1}$ elements, denoted as ${\wt{U}_1, \dots, \wt{U}_L}$ \footnote{One could have also sampled from the $r$-qubit Clifford group which is an exact $3$-design \cite{webb2016cliffordgroupformsunitary, Mele2024introductiontohaar}.}. The operation $\mathbb{I}_M \otimes U_r = \oplus_{i=1}^{M} U_r$ then applies the same unitary $U_r$ independently to each block of the post-measurement state $\bar{\rho}$.  

Let $\rho''$ denote the resulting state after this application. Then, averaging over the randomness in $U_r$, we find
\begin{align}
\mathbb{E}_{U_r\sim\mc{P}(r)}[\rho'']&=\bigoplus_{j=1}^{M} \left[\dfrac{1}{L}\sum_{\ell=1}^{L} \tilde{U}_{\ell}\rho_{jj} \tilde{U}^{\dag}_{\ell}\right]\nonumber\\
&=\bigoplus_{j=1}^{M} \int_{U\sim\mathsf{U}(2^r)} U \rho_{jj} U^{\dag}~\mathrm{d}\mu(U)
=\wt{\rho},
\label{eq:expected-state-algorithm-1}
\end{align}
where the second equality follows because the Pauli group is a unitary 1-design.

As discussed earlier, applying the global unitary $U_{\mathrm{gl}}$ to the coarse-grained state $\wt{\rho}$ produces the passive state associated with $\wt{\rho}$. Therefore, by linearity of expectation, the average output of Step 2(b) of Protocol~\ref{prot1} is
$$
\mathbb{E}[\rho']=\ug~\mathbb{E}_{U_r\sim\mc{P}(r)}[\rho'']~\ugdag=U_{\mathrm{gl}}~\wt{\rho}~U_{\mathrm{gl}}^{\dag}.
$$
In Protocol~\ref{prot1}, Step 1 involves performing a POVM measurement of the Hamiltonian $H$ on the state $\rho$. The measurement outcome is a random variable taking values in the interval $[-\|H\|, +\|H\|]$, with expectation value given by $\Tr[H \rho]$. Similarly, the outcome of Step 2(c), which measures $-H$ on the output state $\rho'$, is also a random variable, with expected value
$$
\Tr[-H~\mathbb{E}[\rho']]=\Tr\left[-H~\ug \mathbb{E}_{U\sim \mc{P}(r)} [\rho'' ]\ugdag \right]=-\Tr[H~\ug\wt{\rho}\ug^{\dag}].
$$ 
As both these measurements are performed independently with probability $1/2$, the outcome of the $j^{\mathrm{th}}$ run of Steps 1 and 2 of Protocol \ref{prot1} is a random variable $\mu_j$ such that 
\begin{equation}
\label{eq:ergobs-on-avg}
   2\mathbb{E}[\mu_j]=\Tr[ H(\rho-\ug\wt{\rho}\ugdag)]=\ergobs{\rho}. 
\end{equation}
Since these steps in Protocol \ref{prot1} are executed a total of $T$ times, the outcome is a sequence of~i.i.d.~random variables $\{\mu_j\}_{j=1}^{T}$. Finally, in Step 5, the overall output of our method is twice the mean $\mu$ of these random variables. We find $T$ for which the sample mean $\mu$ is an $\varepsilon$-accurate estimate of the actual expectation value ($\ergobs{\rho}$) via Hoeffding's inequality \cite{Hoeffding1963}, which implies for any $\varepsilon>0$, 
$$\mathrm{Pr}\left[\left|\mu -\ergobs{\rho} \right| \geq \varepsilon \right]\leq 2\exp\left[-T \varepsilon^{2}\norm{H}^{-2}/8 \right].
$$ 
Thus, $\mu$ approximates observational ergotropy to $\varepsilon$-additive accuracy, i.e., \ $|\mu-\ergobs{\rho}|\leq \varepsilon$, with probability at least $1-\delta$, as long as
\begin{align}
\label{eq:samp-work-prot}
T=O\left(\dfrac{\norm{H}^2\log\left(1/\delta\right)}{\varepsilon^2}\right).
\end{align}
Overall, the probability estimation stage (See Lemma \ref{lem:sample-prob-with-error}) and the work estimation stage combined require 
\begin{align}
    \wt{T}+T=O\left(\dfrac{\|H\|^2\log(M/\delta)}{\varepsilon^2}\right)
\end{align} samples (equivalently, as many measurements) of $\rho$ and succeeds with probability $(1-\delta)^2$. Note that the overall sample complexity depends only logarithmically on the dimension of the system.

\begin{figure}[ht!]
\centering
\includegraphics[width=0.45\textwidth]{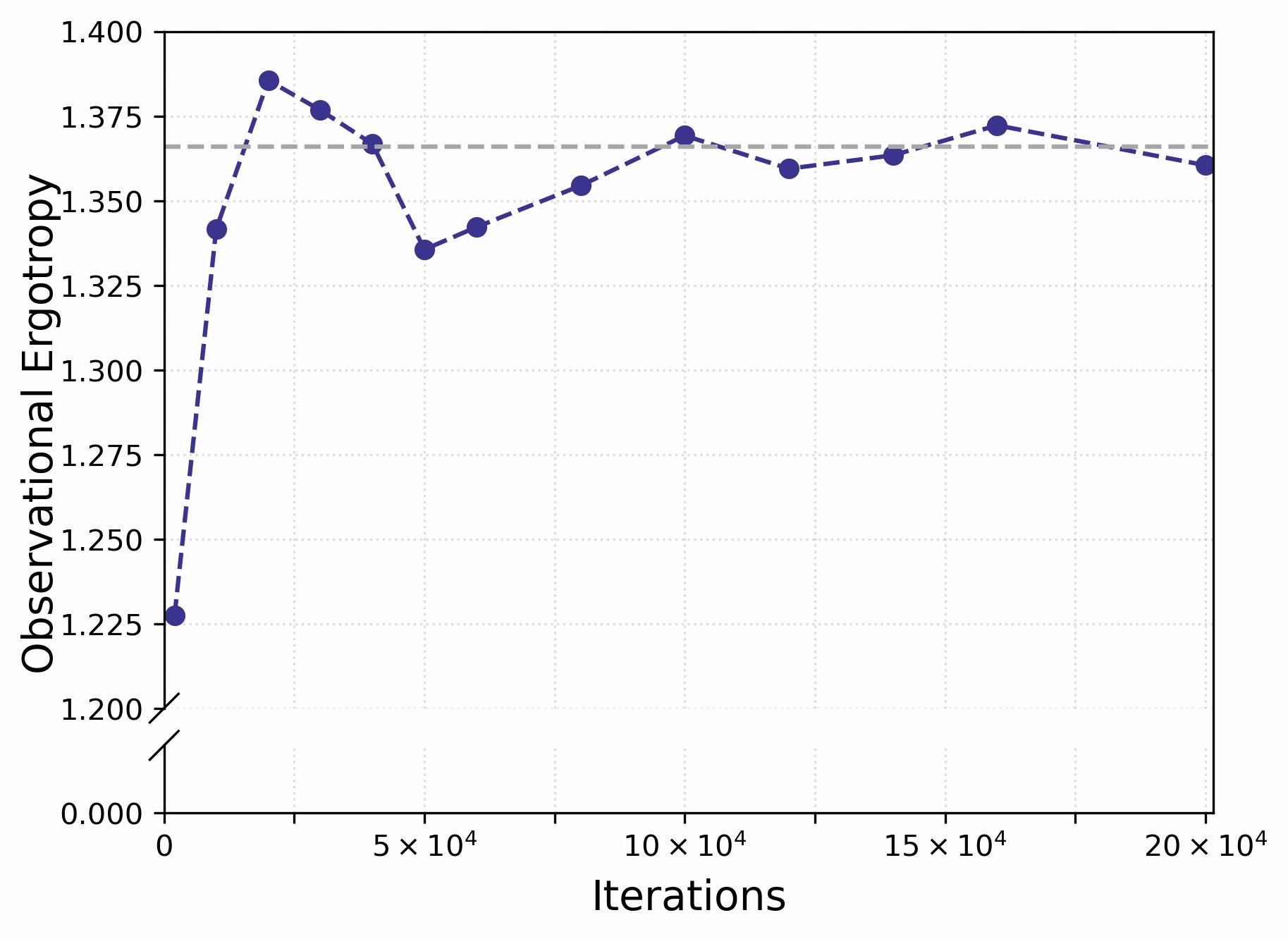}
\caption{Work extraction from translationally invariant Heisenberg XXX model. The Hamiltonian for $n$ spins is given by $ H_{n} = - \sum_{j=1}^n\sigma_z^{(j)} \otimes \sigma_z^{(j+1)} -  \sum_{j=1}^n\sigma_x^{(j)}$ (with the periodic boundary condition), where  $\sigma_z^{(j)}$ and $\sigma_x^{(j)}$ for Pauli-$Z$ and Pauli-$X$ matrices for the $j$th spin, respectively. For $n=3$, let the spins be in the state $\rho = \ket{11}\bra{11}\otimes \mathbb{I}_2/4+ \frac{1}{2}\ket{\psi}\bra{\psi}$, where $\ket{\psi} = \frac{1}{\sqrt{3}}\left(-\ket{011} + \ket{100} + \ket{101}\right)$. We have $\erg{\rho}\approx 2.53$ while $\ergobs{\rho}\approx 1.37$. To estimate the latter quantity, we perform projective measurement $\{L_i\}_{i=1}^4$ with $L_1=\ket{00}\bra{00}\otimes \mathbb{I}_4$, $L_2=\ket{01}\bra{01}\otimes \mathbb{I}_4$, $L_3=\ket{10}\bra{10}\otimes \mathbb{I}_4$, and $L_4=\ket{11}\bra{11}\otimes \mathbb{I}_4$. The figure shows how Protocol \ref{prot1} approximates $\ergobs{\rho}$ as a function of the number of samples $T$. For $T$ as in Eq.~\eqref{eq:samp-work-prot}, we obtain an $\varepsilon$-accurate estimate of $\ergobs{\rho}$ for  $\varepsilon=10^{-2}$, and failure probability $\delta=10^{-4}$.}
\label{fg.samp-spin-main-text}
\end{figure}

In Fig.~\ref{fg.samp-spin-main-text}, we numerically demonstrate how Protocol \ref{prot1} estimates the observational ergotropy. We plot this quantity as a function of samples of the unknown state. We consider a 3-qubit system described by the Heisenberg XXX model Hamiltonian \cite{sachdev2011quantum}. For $\varepsilon=10^{-2}$, $\delta=10^{-4}$, observe that Protocol \ref{prot1} requires roughly $1.5/\varepsilon^2$ samples to estimate the desired quantity. To help illustrate our findings, we provide other examples in Sec.~\ref{append-examples1} and Sec.~\ref{append-examples2} of the Appendix.

Our method can also be seen as a randomized quantum algorithm to estimate the amount of extractable work from an unknown state. Since the projectors $P_i$ are easily implementable, both stages can be implemented on a quantum computer requiring $\wt{T}+T$ independent iterations. Interestingly, the quantum algorithm is qubit efficient, i.e.,\ it requires no ancilla qubits and can potentially be implemented in the intermediate term. If $\tau_{g}$ is the circuit depth of the global unitary $\ug$, then the circuit depth of each run is simply $O(\tau_g)$, as it requires constant circuit depth to implement a randomly sampled unitary from the $r$-qubit Pauli group (in Step 2(b)). The precise depth of constructing $\ug$ is instance-specific: it depends on $\rho$ and the choice of the projective measurement $\{P_i\}$, which determines the amount of coarse-graining. Note that this unitary implements (i) a basis change from the basis spanned by the projectors to the eigenbasis of the Hamiltonian, followed by (ii) a sequence of permutations required to arrange the estimated $\wt{p_i}$ in non-increasing order. It is possible to implement (i) as $H$ is completely known, while the set of projectors $\{P_i\}$ is chosen by the user. Also, (ii) can be efficiently executed \cite{Quek2024multivariatetrace}, with the gate depth depending on the number of permutations required to be implemented.

Our algorithm can be made more efficient in physically relevant settings. Observe that our scheme performs a measurement of $H$ on $\rho$, which might be difficult to implement. However, consider any $H=\sum_{j} h_j \Lambda_j$, which is a linear combination of $n$-qubit Pauli strings $\Lambda_j$ with total weight $h=\sum_{j}|h_j|$. For such Hamiltonians, ubiquitous across physics \cite{sachdev2011quantum}, instead of measuring $H$ ($-H$) in Step 1 (Step 2(c)) in each run of our algorithm, we can independently sample a Pauli term $\Lambda$ according to the ensemble $\{h_i/h,~\Lambda_i\}_i$ and measure $\Lambda$, and $-\Lambda$, respectively. It suffices to adopt this simple randomized measurement scheme, which enables the measurement of a single local Pauli operator in each run. In Sec.~\ref{subsec:pauli-measurement}, we prove that the modified protocol (See Protocol \ref{prot1mod} in Sec.~\ref{subsec:work-multiple-copy}) outputs $\mu$ which approximates $\ergobs{\rho}$ to $\varepsilon$-additive accuracy, with probability at least~$1-\delta$ as long as $T=O(h^2\log(1/\delta)/\varepsilon^2)$.

\subsection{Robustness of observational ergotropy}
\label{sec:robustness-ergobs}
As described earlier, estimating observational ergotropy involves a two-step process. Given an unknown quantum state $\rho$ and a set of orthogonal projectors $\{P_i\}_{i=1}^{M}$, satisfying $\sum_{i} P_i=\mathbb{I}_N$, our algorithm first estimates $p_i=\Tr[P_i\rho]$, for all $i\in [1,M]$, up to $\varepsilon$-additive accuracy. That is, we obtain estimates $\widetilde{p}_i$ such that 
$$|\widetilde{p}_i-p_i|\leq \max_{i\in[1,M]}|\widetilde{p}_i-p_i|=: \varepsilon,$$ 
for all $i\in [1,M]$. This estimation step is critical because the construction of the global unitary $U_{\mathrm{gl}}$ relies on the correct ordering of the probabilities $\{p_i\}$. In fact, these probabilities need to be arranged in a nonincreasing order. Therefore, to construct the correct $U_{\mathrm{gl}}$,  it is essential that the estimation accuracy preserves this ordering: for all $i,j \in [1, M]$, if $p_i\geq p_j$, then it must also hold that $\wt{p_i}\geq \wt{p_j}$. In Section~\ref{subsec:work-multiple-copy}, we show that this ordering is guaranteed to be preserved if the estimation error satisfies $\varepsilon<\Delta/2$, where $\Delta= \min_{i,j}\{|p_i-p_j|: p_i\neq p_j\}$. However, in realistic black-box scenarios, the true probabilities $\{p_i\}$ are unknown, and so is the minimum gap $\Delta$. As a result, the experimentalist must choose a precision $\varepsilon$ without knowing whether it suffices to preserve the correct ordering. If $\varepsilon$ is too large, the estimated ordering may differ from the true ordering, resulting in the use of an incorrect global unitary $\wt{U}_{\mathrm{gl}}$ instead of $U_{\mathrm{gl}}$. In such cases, Protocol \ref{prot1} estimates a \textit{guessed observational ergotropy}, defined as
\begin{equation}
    \label{eq:guess-erg-def}
    \wt{\mathrm{Erg}}_{\mathrm{obs}}(\rho)=\Tr\left[H\left(\rho-\wt{U}_{\mathrm{gl}}~\wt{\rho}~\wt{U}_{\mathrm{gl}}^{\dagger} \right)\right],
\end{equation}
where $\wt{\rho}$ is the coarse-grained state (see Eq. \eqref{eq:randomized-state}). In what follows, we derive an upper bound on the deviation between the guessed and the true values of observational ergotropy.

To this end, let us redefine the coarse-grained state 
$$
\wt{\rho}=\bigoplus_{i=1}^{M}\dfrac{p_iP_i}{d},
$$
such that $p_1\geq p_2\geq \cdots \geq p_M$. Here $p_i$ are the \emph{true} block probabilities, arranged in a non-decreasing order. Suppose the eigenvalues of the Hamiltonian (ordered nondecreasingly) are $E_1\leq E_2\leq \cdots \leq E_N$. Then, we partition this ordered list of eigenvalues into $M$ contiguous blocks $A_1, A_2, \cdots, A_M$, each of size $d=2^r$, such that
$$
A_j=\{E_{(j-1)d+1}, E_{(j-1)d+2}, \cdots, E_{jd}\}.
$$
Each block $A_j$ defines a projector $P_j$ onto the span of the eigenstates corresponding to the eigenvalues in $A_j$. Then, the average energy of block $A_j$ is defined as
\begin{equation}
    \label{eq:block-avg-energy}
    \overline{E}_j=\dfrac{1}{d}\sum_{e\in A_j} e.
\end{equation}
Thus, the set $\{p_j, \overline{E}_j\}$ provides the coarse-grained description of $\rho$ in the energy eigenbasis, which serves as the input for constructing the global unitary $\ug$, and for defining the true observational ergotropy.

Now the experimentalist obtains the \textit{estimated} block probabilities $\wt{p}_1, \cdots, \wt{p}_M$, i.e.\ each $\wt{p}_i$ is associated with the block labeled by $P_i$. Let $\pi$ be the permutation that sorts these estimated probabilities:
$$
\wt{p}_{\pi(1)}\geq \wt{p}_{\pi(2)} \geq \cdots \wt{p}_{\pi(M)}.
$$
Consequently, the global unitary used by the experimentalist is constructed from the sorted estimated probabilities as $\wt{U}_{\mathrm{gl}}$, which maps the eigenvectors of $P_{\pi(i)}$ to the energy block $A_i$. Let $p_{\pi(i)}$ denote the \emph{true} probability of the block $P_{\pi(i)}$ which the experimentalist believes to be the $i$-th largest. Then, we prove the following lemma.

\begin{lemma}
\label{lem:guess-ergobs-ergobs-diff}
    For any density operator $\rho$, let $p_i=\Tr[\rho P_i]$ for $i\in[1,M]$ be the true measurement probabilities and $\{\wt{p}_i\}_{i=1}^M$ be the estimated probabilities. Let $\pi$ be a permutation on $M$ objects such that for $i\in[1,M-1]$, 
    $$\wt{p}_{\pi(i)}\geq \wt{p}_{\pi(i+1)}.$$ 
    Then,
    \begin{align}
    \label{eq:bound-good}
     \left|\wt{\mathrm{Erg}}_{\mathrm{obs}}(\rho)-\ergobs{\rho}\right|= \left|\sum_{j=1}^M \left(p_j-p_{\pi(j)}\right)\overline{E}_j\right|\leq \|H\|\sum_{j=1}^M\left|p_j-p_{\pi(j)}\right|,
    \end{align}
where $\overline{E}_j$ is defined as in Eq.~\eqref{eq:block-avg-energy}.
\end{lemma}
The proof of Lemma~\ref{lem:guess-ergobs-ergobs-diff} is provided in Sec.~\ref{sub:proofs-lems}, while in Appendix \ref{append:ex-rob}, we provide an example that complements our theoretical findings. The lemma illustrates how inaccuracies in estimating the block probabilities directly propagate into deviations of the \textit{guessed observational ergotropy} from the true value. Specifically, the deviation is governed solely by the mismatch between the true probabilities and their permuted counterparts under the estimated ordering, weighted by the block energies. This makes clear that the “penalty” for implementing an incorrect global unitary is controlled: the discrepancy can never exceed the operator norm of the Hamiltonian multiplied by the total probability mass misassigned due to sorting errors. Consequently, even when the estimation precision $\varepsilon$ is too coarse to guarantee the correct ordering, one can still obtain a rigorous upper bound on the possible loss (or gain) in the extracted work. Importantly, the lemma also shows that the \textit{guessed observational ergotropy} is not necessarily smaller than the true value. Instead, it provides an explicit quantification of the worst‐case deviation that an experimentalist should expect, thereby shedding light on the robustness of our protocol under imperfect probability estimation.

\section{Preliminaries}
\label{prelim}
In this section, we introduce standard notation, definitions, and facts that have been utilized to derive our main results. 

Recall that $N=2^n$, the Hilbert space of $n$ qubits is denoted by $\mathcal{H}_N$ and the set of density operators, trace one positive semi-definite matrices, on $\mathcal{H}_N$ is denoted by $\mathcal{D}_N$. Further, define $[s]=\{1,\cdots,s\}$. The Hamiltonian of the $n$ qubit system, denoted by $ {H}$, satisfies
\begin{align}
\label{eq:full-hamiltonian1}
     {H} = \sum_{a=1}^s E_a\sum_{\kappa=1}^{g_a}\ket{E_a;\kappa}\bra{E_a;\kappa}:= \sum_{a=1}^s E_a  {Q}_a,
\end{align}
where $ {Q}_a:=\sum_{\kappa=1}^{g_a}\ket{E_a;\kappa}\bra{E_a;\kappa}$ is the projector onto the eigensubspace of energy $E_a$, $g_a$ is the degeneracy of the energy eigenvalue $E_a$, and $E_a < E_{a+1}$ for all $a\in[s-1]$. $s$ is the total number of distinct eigenvalues of $ {H}$ and $\sum_{a=1}^s g_a=N$. Let us further define an ordering on the eigenstates, for convenience, as follows: $\ket{E_1;1}\leq \cdots\leq \ket{E_1;g_1}\leq  \cdots\leq \ket{E_s;1}\leq\cdots\leq \ket{E_s;g_s}$. Denote by $\left(\ket{\phi_1},\cdots, \ket{\phi_{g_1}},\cdots, \ket{\phi_{N-g_s+1}},\cdots,\ket{\phi_N}\right)$ the ordered tuple of eigenvectors arranged in order of nondecreasing energy. Here we have partitioned $N$ as $\{g_1,\cdots, g_s\}$ such that $\sum_{a=1}^s g_a=N$. In particular, the set $$\left\{\ket{\phi_{1+\sum_{a=1}^{l-1} g_a}} ,\cdots,\ket{\phi_{\sum_{a=1}^{l} g_a}}\right\}$$  spans the eigenspace of energy $E_l$, where $l\in[s]$.

\medskip
\noindent
{\bf Ergotropy:~}Let us consider a quantum system of $n$ qubits with the Hamiltonian $H$ (Eq.~\eqref{eq:full-hamiltonian1}) in an arbitrary state $\rho\in \mathcal{D}_N$. Note that the extractable work is a path/process-dependent quantity; the extractable work from any quantum system can be defined in multiple ways depending on the quantum operations used for work extraction (see, e.g., Refs .~\cite {Lenard78, Michal2013, Narasimhachar2015, Singh2019}). Ergotropy captures the notion of extractable work under unitary operations. As mentioned previously, for a quantum system in a state $\rho$ \cite{Lenard78, Allahverdyan2004}, 
\begin{align}
\label{eq:erg-def-first1}
    \erg{\rho}&=\max_{U\in \mathsf{U}(N)}\Tr\left[H\left(\rho-U\rho U^{\dag}\right)\right],
\end{align}
where $\mathsf{U}(N)$ is the group of $N\times N$ unitary matrices. Let the spectral decomposition of $\rho$ be given by $\rho=\sum_{i=1}^N\lambda_i\ket{\lambda_i}\bra{\lambda_i}$, where $\lambda_i \geq 0$ and $\sum_{i=1}^N\lambda_i=1$. Let $\pi$ be the permutation of $N$ objects such that $\lambda_{\pi(i)}\geq \lambda_{\pi(j)}$ for all $1\leq i\leq j\leq N$. Then the unitary matrix $\vg \in \mathsf{U}(N)$ achieving the maximum in Eq. \eqref{eq:erg-def-first1} can be written as
\begin{align}
\label{eq:one-way-unitary}
    &\vg = \sum_{a=1}^s\sum_{\kappa=1}^{g_a}\ket{\phi_{\kappa+\sum_{l=1}^{a-1}g_l}}\bra{\lambda_{\pi\left({\kappa+\sum_{l=1}^{a-1}g_l}\right)}}.
\end{align}
Thus, ergotropy of the state $\rho$ can be written as
\begin{align}
   \erg{\rho}&=\Tr\left[H\left(\rho-\vg \rho \vgdag\right)\right].
\end{align}
For a nondegenerate Hamiltonian $H=\sum_{i=1}^N E_i\ket{E_i}\bra{E_i}$ (with $E_1 <\cdots<E_N$), the above expression simplifies to
\begin{align}
    \erg{\rho}&=\sum_{i,j=1}^N E_i \lambda_j\left(|\bra{\lambda_j}E_i\rangle|^2 - \delta_{j,\pi(i)}\right).
\end{align}

\medskip
\noindent
{\bf Observational ergotropy:~}For a physical system described by Hamiltonian $H$ in $\mathcal{H}_N$, the observational ergotropy of a state $\rho\in \mathcal{D}_N$ is defined with respect to a measurement setting. Consider the projective measurement $\{P_i\}_{i=1}^M$ described by the projectors $P_i$ such that $\sum_{i=1}^M P_i=\mathbb{I}_N$, and  $\mathcal{H}_i=P_i\mathcal{H}_NP_i$, for all $i\in [M]$. Then, the observational ergotropy of $\rho$ is defined as
\begin{align}
    \ergobs{\rho} &= \Tr\left[H \rho\right] - \min_{U\in \mathsf{U}(N)}\Tr\left[H U~ \wt{\rho} ~U^{\dag}\right]\nonumber\\
    &= \Tr\left[H \rho\right] - \Tr\left[H \ug~\wt{\rho}~\ugdag\right],
\end{align}
where $\wt{\rho}=\sum_{i=1}^M p_i P_i/\Tr[P_i]$, with $p_i=\Tr[\rho P_i]$ for $i\in[M]$. If $\Tr[P_i]$ is the same for all $i$, $\ug$ is a global unitary that implements (i) a basis change from $\oplus_{i=1}^{M} \mathcal{H}_i$ to the eigenbasis of $H$, followed by (ii) a permutation that orders the $p_i$'s in non-increasing order. We refer the readers to Refs.~\cite{Safranek2019, Safranek2021, vsafranek2023work} for bounds on observational ergotropy and other details.

\medskip
\noindent
\textbf{Induced measures on the space of mixed states:~}The notion of random states requires a probability measure to be associated with the set of states. For an $N$-dimensional system with Hilbert space $\mc{H}_N$, a random pure state $\ket{\psi}$ may be written as $\ket{\psi} = U\ket{\psi_0}\in \mc{H}_N$, where $\ket{\psi_0}$ is a fixed pure state and $U$ is an $N\times N$ unitary matrix chosen uniformly from the unique Haar measure $\mu_{\mathrm{Haar}}(U)$ on the unitary group $\mathsf{U}(N)$ (see e.g. Refs. \cite{Jones1990, Zyczkowski2001}).
However, the set $\mc{D}_N$ of density matrices (nonnegative and trace one $N\times N$ matrices) does not admit any unique probability measure on it, and in fact, there exist several inequivalent measures on $\mc{D}_N$. For example, one may generate a random density matrix $\rho \in \mc{D}_N$ by partial tracing the purifications $\ket{\psi}\in\mc{H}_N\otimes \mc{H}_{N'}$  chosen uniformly from the unique Haar measure on the pure states in $\mc{H}_N\otimes \mc{H}_{N'}$, where $N'\geq N$. This procedure induces the probability measure $d\nu_{N,N'} (\Lambda)$ on the eigenvalues
$\{\lambda_1,\ldots,\lambda_N\}$ of $\rho$, which is given by \cite{Zyczkowski2001, Zyczkowski2003}
\begin{eqnarray}
d\nu_{N,N'} (\Lambda) := C_{N,N'}~
\delta\left(1-\sum^N_{j=1}\lambda_j\right)\abs{\Delta(\lambda)}^{2}\prod^N_{j=1}\lambda^{N'-N}_j d\lambda_j,
\end{eqnarray}
where $C_{N,N'}$ is the normalization constant and  $\Delta(\lambda)=\prod_{1\leqslant i<j\leqslant
N}(\lambda_i-\lambda_j)$. The eigenvectors of $\rho$ are distributed via the unique Haar measure on $\mathsf{U}(N)$. The normalized probability measure $\mu_{N,N'}$ over $\mc{D}_N$ is then given by
\begin{eqnarray}
\label{eq:meas-purification}
d\mu_{N,N'}(\rho) =
d\nu_{N,N'}(\Lambda)\times d\mu_{\mathrm{Haar}}(U),
\end{eqnarray}
where $\rho=U\Lambda U^{\dagger}$, $d\nu_{N,N'}(\Lambda)\equiv d\nu_{N,N'}(\lambda_1,\ldots,\lambda_N)$.

Another way to obtain a probability measure on $\mc{D}_N$ is to note that any density matrix $\rho$, using the spectral decomposition theorem, can be written as $\rho=U\Lambda U^{\dagger}$, where $U$ is an $N\times N$ unitary matrix and $\Lambda=\mathrm{diag}\left(\lambda_1,\ldots,\lambda_N\right)$ is a diagonal matrix with nonnegative real entries that sum to one. Then assuming that the distributions of eigenvalues and eigenvectors of $\rho\in \mc{D}_N$ are independent, one may propose a probability measure $\mu$ on $\mc{D}_N$ as $\mu(\rho)=\nu(\Lambda)\times\mu_{\mathrm{Haar}}(U)$, where the measure $\mu_{\mathrm{Haar}}(U)$ is the unique Haar measure on $\mathsf{U}(N)$ and measure $\nu(\Lambda)$ defines the distribution of eigenvalues (see e.g. Refs. \cite{Hall1998, Zyczkowski2001}). The measure $\nu(\Lambda)$ is not unique and one may pick it based on different metrics on the Hilbert space \cite{Slater1999}. In particular, for Hilbert-Schmidt metric, $d\nu(\Lambda)$ coincides with $d\nu_{N,N}(\Lambda)$ given in Eq. \eqref{eq:meas-purification} \cite{Zyczkowski2001}. In this work, we will focus only on the partial trace induced measure on the set of density matrices.

\medskip
\noindent
{\bf Lipschitz continuity:~} Let $X$ and $Y$ be two metric spaces with metrics $d_X$ and $d_Y$, respectively and $h : X \mapsto Y$ be a function on $X$. Then $h$ is said to be a Lipschitz continuous function with the Lipschitz constant $L$ if for any $x,x'\in X$
\begin{align}
    d_Y\left(h(x)-h(x')\right) \leq L~ d_X(x,x').
\end{align}
Note that any other constant $L'\geq L$ is also a valid Lipschitz constant \cite{Milman1986, Ledoux2005, Searcoid2007}. 
%

\medskip
\noindent
{\bf Concentration of measure phenomenon:~}Certain smooth functions defined over measurable vector spaces are known to take values close to their
average values almost surely \cite{Ledoux2005}. This phenomenon is collectively referred to as the concentration of measure phenomenon and is expressed via various inequalities bounding the probability of departure of the functional values from the average value. Logarithmic Sobolev inequalities together with the Herbst argument (also called the “entropy method”) are very general techniques to prove such inequalities \cite{Raginsky2013}. A particular form of the concentration inequality suitable for our exposition is stated as the following lemma \cite{Milman1986, Ledoux2005, Hayden2006, Meckes2013}.
\begin{lemma}[Concentration of measure phenomenon]
\label{lem:levy's lemma}
Let $\mathsf{U}(k)$ be the group of $k\times k$ unitary matrices equipped with the Hilbert-Schmidt metric. Let $h:\mathsf{U}(k)\mapsto \mathbb{R}$ be a real valued Lipschitz continuous function with a Lipschitz constant $L$. Then for any $\gamma>0$, we have 
    \begin{align}
    \Pr\left(h(U)>\mathbb{E}_{\mathsf{U}(k)}[h(U)]+ \gamma \right)\leq \exp\left[-\frac{k\gamma^2}{12 L^2}\right].
\end{align}
\end{lemma}

\section{Methods}
\label{sec:methods}
Having outlined the necessary preliminaries, we are now in a position to describe the techniques used to derive the main results in Sec.~\ref{sec:main-results}. We begin with a detailed derivation of Theorem \ref{thm:impossibility-work}.

\subsection{Proof of Theorem \ref{thm:impossibility-work}: Concentration of ergotropy}
\label{append-average-ergotropy}
The central idea is to show that the ergotropy of a full rank random quantum state $\rho$ concentrates around its average ergotropy, which is zero in the asymptotic limit. We develop this proof via a sequence of non-trivial steps. First, we prove an upper bound on the average ergotropy, and then show that the ergotropy functional is Lipschitz continuous. Finally, we leverage the concentration of measure phenomena to show that $\mathrm{Erg}(\rho)$ is close to its average with a high probability.

\medskip
\noindent
\textbf{Average ergotropy:~}Let $\rho\equiv \rho(U)$ be a random state in $\mc{D}_N$ obtained via partial tracing Haar random pure states on $\mc{H}_N\otimes \mc{H}_{N'}$. In particular, $\rho(U) = \Tr_B\left[U\ket{\psi_0}\bra{\psi_0}_{AB}U^{\dag}\right]$, where $\ket{\psi_0}_{AB}$ is a fixed state. Let $\{\lambda_i(U)\}_{i=1}^N$ be the eigenvalues of $\rho(U)$. Further, let $\pi$ be a permutation on $N$ objects such that $\lambda_{\pi(1)}(U)\geq \cdots\geq \lambda_{\pi(N)}(U)>0$ while $E_1\leq \cdots\leq E_N$. The ergotropy of $\rho(U)$ is given by $\erg{\rho(U)}=\max_{V\in\mathsf{U}(N)}\Tr\left[H\left(\rho(U)-V \rho(U) V^{\dag}\right)\right]$. We have
\begin{align}
   \erg{\rho(U)}&=\Tr\left[H\rho(U)\right]-\sum_{i=1}^N E_i\lambda_{\pi(i)}(U).
\end{align}
Since $\rho(U)$ is a random state, $\erg{\rho(U)}$ is a random variable. In the following lemma, we estimate the expectation value $\mathbb{E}_{\mathsf{U}(NN')}\left[\erg{\rho(U)}\right]$ of $\erg{\rho(U)}$. 
\begin{lemma}
\label{lem:avg-ergotropy}
    The average ergotropy for a full rank random quantum state sampled from the distribution $\mu_{N,N'}(\rho)$ (Eq. \eqref{eq:meas-purification}) is bounded from above by $\norm{H}\,\sqrt{N/N'}$.
\end{lemma}
\begin{proof}
Let us denote $\mathbb{E}_{\mathsf{U}(NN')}$ by $\mathbb{E}$ for the ease of notation. Using $\mathbb{E}[\rho(U)]=\mathbb{I}_N/N$ (see e.g. \cite{Zyczkowski2001}), we have
\begin{align}
\mathbb{E}\left[\erg{\rho(U)}\right]&=\sum_{i=1}^N  E_i\, \mathbb{E} \left[\delta_i(U)\right],
\end{align}
where $\delta_i(U)=1/N-\lambda_{\pi(i)}(U)$. From Cauchy-Schwarz inequality, we have
\begin{align}
\mathbb{E}\left[\erg{\rho(U)}\right]&\leq \sqrt{\sum_{i=1}^N  E_i^2}\,\sqrt{\sum_{i=1}^N\left(\mathbb{E} \left[\delta_i(U)\right]\right)^2}\nonumber\\
&\leq \sqrt{N}\norm{H}\sqrt{\sum_{i=1}^N\left(\mathbb{E} \left[\delta_i(U)\right]\right)^2}\nonumber\\
&\leq \sqrt{N}\norm{H}\sqrt{\mathbb{E} \left[\sum_{i=1}^N \delta^2_i(U)\right]},
\end{align}
where, in the last line, we applied Jensen's inequality. Noting that $\mathbb{E} \left[\sum_{i=1}^N \delta^2_i(U)\right] = \mathbb{E}\left[\sum_{i=1}^N\lambda^2_{\pi(i)}(U)\right] - 1/N = \mathbb{E}\left[\Tr\left[\rho(U)^2\right]\right]-1/N$ and  $\mathbb{E}\left[\Tr(\rho^2)\right] = \left(N + N'\right)/\left(N N' + 1\right)$ \cite{Zyczkowski2001}, we obtain
\begin{align}
\mathbb{E}\left[\erg{\rho(U)}\right]
&\leq \norm{H}\sqrt{\frac{N}{N'}}.
\end{align}
This concludes the proof of the lemma.
\end{proof}

\medskip
\noindent
\textbf{Lipschitz continuity of Ergotropy:~}Next, we prove that $\erg{\rho(U)}$ is a Lipschitz continuous function. Note that this was recently proven in Ref.~\cite{hovhannisyan2024concentration}. We prove a slightly different version of the same result here via the following Lemma:
\begin{lemma}
\label{lem:Lipschitz-ergotropy}
    The ergotropy of quantum states is a Lipschitz continuous function with Lipschitz constant $L=4 \norm{H}$. More specifically, let $\rho (U), \rho(V) \in \mc{D}_N$ such that $\rho(U)=\Tr_B\left[\ket{\psi(U)}\bra{\psi(U)}_{AB}\right]$ and $\rho(V)=\Tr_B\left[\ket{\psi(V)}\bra{\psi(V)}_{AB}\right]$, where $B$ is an $N'$-dimensional system and $\ket{\psi(U)}_{AB}=U\ket{\psi_0}$. Here $\ket{\psi_0}$ is a fixed normalized state and  $U$ and $V$ are unitaries. Then we have
\begin{align}
\left|\erg{\rho(U)}-\erg{\rho(V)}\right|\leq 4\norm{H}~\norm{U-V}_2,
\end{align}
where $\norm{\cdot}_2$ is the Hilbert-Schmidt norm.
\end{lemma}
\begin{proof}
Let $\{\lambda_i(\rho(U))\}$ and $\{\lambda_i(\rho(V))\}$ be the spectrum the of states $\rho(U)$ and $\rho(V)$. Let $\lambda_1^{\downarrow}(\rho(U)) \geq \cdots\geq \lambda_N^{\downarrow}(\rho(U))$ and $\lambda_1^{\downarrow}(\rho(V)) \geq \cdots\geq \lambda_N^{\downarrow}(\rho(V))$. Then the ergotropy of $\rho(U)$ and $\rho(V)$ is given, respectively, by
\begin{align}
    &\mathrm{Erg}(\rho(U)) = \mathrm{Tr}\left[H\rho(U)\right] - \sum_{i=1}^N E_i \lambda_i^{\downarrow}(\rho(U)),\nonumber\\
    &\mathrm{Erg}(\rho(V)) = \mathrm{Tr}\left[H\rho(V)\right] - \sum_{i=1}^N E_i \lambda_i^{\downarrow}(\rho(V)),\nonumber
\end{align}
where $H = \sum_{i=1}^N E_i \,\ket{E_i}\bra{E_i}$ with increasing energies $E_1 \leq \cdots \leq E_N$. Then
\begin{align}
\label{eq:b11}
\left|\mathrm{Erg}(\rho(U)) - \mathrm{Erg}(\rho(V))\right| &\leq \left|\Tr(H(\rho(U) - \rho(V)))\right| + \left|\sum_i E_i (\lambda_i^{\downarrow}(\rho(U)) - \lambda_i^{\downarrow}(\rho(V)))\right| \nonumber\\
&\leq \norm{H} \, \norm{\rho(U) - \rho(V)}_1 + \norm{H} \, \norm{\lambda^{\downarrow}(\rho(U)) - \lambda^{\downarrow}(\rho(V))}_1,
\end{align}
where $\lambda^{\downarrow}(\rho(U))$ is the vector with components $\lambda_1^{\downarrow}(\rho(U)) \geq \cdots\geq \lambda_N^{\downarrow}(\rho(U))$. Here $\norm{\cdot}_1$ is the trace norm. Applying the trace-norm version of Lidskii's inequality \cite{Bhatia1997}, we have
\begin{align}
\label{eq:b12}
\norm{\lambda^{\downarrow}(\rho(U)) - \lambda^{\downarrow}(\rho(V))}_1 \leq \norm{\rho(U) - \rho(V)}_1.
\end{align}
Combining Eqs. \eqref{eq:b11} and \eqref{eq:b12}, we have
\begin{align*}
|\mathrm{Erg}(\rho(U)) - \mathrm{Erg}(\rho(V))| &\leq 2\norm{H}\,\norm{\rho(U) - \rho(V)}_1\nonumber\\
&\leq 2\norm{H}\,\norm{\ket{\psi(U)}\bra{\psi(U)} - \ket{\psi(V)}\bra{\psi(V)}}_1,
\end{align*}
where in the last line we used the contractivity of the trace norm under partial trace. Note that
\begin{align}\label{2norm}
\norm{\ket{\psi(U)}\bra{\psi(U)} - \ket{\psi(V)}\bra{\psi(V)}}_1 = 2\sqrt{1 - |\bra{\psi(U)}\psi(V)\rangle|^2} \leq 2 \norm{\ket{\psi(U)} - \ket{\psi(V)}}_2.
\end{align} 
Using $\norm{\ket{\psi(U)} - \ket{\psi(V)}}_2 = \norm{(U - V)\ket{\psi_0}}_2 \leq \norm{U - V}_2$, we get
\begin{align*}
    \left|\mathrm{Erg}(\rho(U)) - \mathrm{Erg}(\rho(V))\right| \leq 4\norm{H} \, \norm{U - V}_2.
\end{align*}
This concludes the proof of the lemma.
\end{proof}
Finally, we prove the main result of this section.
\ \\
\textit{Proof of Theorem~\ref{thm:impossibility-work}}.-- From Lemma \ref{lem:avg-ergotropy}, we have that $\mathbb{E}[\mathrm{Erg}(\rho(U))]\leq \norm{H}\,\sqrt{\frac{N}{N'}}$. Then for $\gamma> 2\norm{H}\,\sqrt{\frac{N}{N'}}$, we have $\gamma/2 > \mathbb{E}[\mathrm{Erg}(\rho(U))]$. Now note that 
\begin{align}
\label{eq:bound-a}
    \Pr\left[\mathrm{Erg}(\rho(U))>\gamma \right] \leq  \Pr\left[\mathrm{Erg}(\rho(U))>\gamma/2 + \mathbb{E}\left[\mathrm{Erg}(\rho(U))\right]\right].
\end{align}
From Lemma \ref{lem:levy's lemma} and Lemma \ref{lem:Lipschitz-ergotropy}, we have
\begin{align}
\label{eq:bound-b}
    \Pr\left[\mathrm{Erg}(\rho(U))>\gamma/2 + \mathbb{E}\left[\mathrm{Erg}(\rho(U))\right]\right] \leq \exp\left[-\frac{NN'\gamma^2}{768 \,\norm{H}^2}\right].
\end{align}
Using Eq. \eqref{eq:bound-b} in Eq. \eqref{eq:bound-a}, we complete the proof of the theorem.

\subsection{Extracting work from multiple copies of $\rho$: proof of Lemma \ref{lem:sample-prob-with-error}}
\label{subsec:proof-multiple-copies}
In this section, we prove Lemma \ref{lem:sample-prob-with-error}. This requires obtaining the number of samples $\wt{T}$ required to estimate the probabilities $\wt{p}_i$ to arbitrary additive precision. Furthermore, as discussed in Sec.~\ref{subsec:work-multiple-copy}, the appropriate choice of the precision parameter is also crucial: ordering the estimated probabilities $\wt{p}_i$'s should not alter the original ordering of the $p_i$'s. This affects the number of samples $\wt{T}$. We first obtain $\wt{T}$ as a function of any arbitrary $\varepsilon >0$, and then substitute the appropriate value of $\varepsilon$ based on our requirement.

\medskip
\noindent
\textbf{Probability estimation:~}Let us consider independent and identically distributed (i.i.d.) random variables $Y_1, \cdots Y_{\wt{T}}$ which satisfy $a_j\leq Y_j\leq b_j$. Let, for all $j\in\left[\wt{T}\right]$, $\mathbb{E}\left[Y_j\right]:=\mathbb{E}\left[Y\right]$. For such i.i.d. random variables, for all $\varepsilon >0$, Hoeffding's inequality \cite{Hoeffding1963, Boucheron2013} states that
\begin{align}
    \Pr\left[\left|\frac{1}{\wt{T}}\sum_{j=1}^{\wt{T}}Y_j-\mathbb{E}[Y]\right| \geq \varepsilon \right]\leq 2\exp\left[-\frac{2\wt{T}^2\varepsilon^2}{\sum_{j=1}^{\wt{T}}(b_j-a_j)^2}\right].
\end{align}
Let $\rho$ be an unknown $n$-qubit quantum state, and we want to estimate probabilities of outcomes of a projective measurement given by $\left\{P_i\right\}_{i=1}^M$ such that $P_i^2=P_i$ and $P_i\geq 0$ for all $i\in[M]$. Further, $\sum_{i=1}^MP_i=\mathbb{I}_N$, where $N=2^n$. Suppose we perform this projective measurement $\wt{T}$ times. At each instance of measurement, the outcome is one of the indices $i\in[M]$. Let us define an indicator function $I_i$, which is one if the measurement outcome is $i$ and zero otherwise. During $j$th repetition of the measurement, we store the value of $I_i$ in the random variable $Y_j^{(i)}$. The expectation value of $Y_j^{(i)}$ for each $j\in\left[\wt{T}\right]$ is given by $\mathbb{E}\left[Y_j^{(i)}\right] = p_i$, where $p_i=\Tr\left[\rho P_i\right]$ for $i\in[M]$. Then, defining $\frac{1}{\wt{T}}\sum_{j=1}^{\wt{T}}Y_j^{(i)}=\wt{p}_i$, from Hoeffding's inequality for i.i.d. random variables $Y_j^{(i)}$, for all $\varepsilon>0$ we have
\begin{align}
    \Pr\left[\left|\wt{p}_i-p_i\right| \geq \varepsilon \right]\leq 2\exp\left[-2\wt{T}\varepsilon^2\right],
\end{align}
where we used the fact that $0\leq Y_j^{(i)} \leq 1$ for each $j\in\left[\wt{T}\right]$ and $i\in[M]$. Now, for $0< \delta < 1$, by choosing $2\exp\left[-2\wt{T}\varepsilon^2\right]\leq \delta/M$ or
\begin{align}
    \wt{T} \geq \frac{1}{2\varepsilon^2}\ln \left(\frac{2M}{\delta}\right),
\end{align}
we have $\Pr\left[\left|\wt{p}_i-p_i\right| \geq \varepsilon\right] \leq \delta/M$. Thus, $\wt{T}$ samples of $\rho$ suffice to estimate $p_i$ within additive error $\varepsilon$ and maximum failure probability $\delta/M$. Now applying union bound on measurement outcomes $i\in [M]$, we have
\begin{align}
    \Pr\left[\bigcup_{i=1}^M G_i\right] \leq \sum_{i=1}^M\Pr\left[ G_i\right] \leq 2M\exp\left[-2\wt{T}\varepsilon^2\right]\leq \delta,
\end{align}
where $G_i$ denotes the event that $\left|\wt{p}_i-p_i\right| \geq \varepsilon$. This implies that $O\left(\log \left(2M/{\delta}\right)/\varepsilon^2\right)$ samples of $\rho$ are sufficient to estimate $p_i$ for all $i\in[M]$ within additive error $\varepsilon$ and maximum failure probability $\delta$.

\medskip
\noindent
{\bf Choosing the error $\mathbf{\varepsilon}$:~}As discussed in Sec.~\ref {subsec:work-multiple-copy}, the estimates $\wt{p_i}$ need to be arranged in non-increasing order for the construction of the global unitary $\ug$. It is thus necessary to ensure that the precision $\varepsilon$ is such that it does not alter the ordering of the original $p_i$'s. In other words, we intend to find error $\varepsilon$ in the estimated probabilities $\left\{\wt{p}_i\right\}_{i=1}^M$ such that $p_i - p_j \geq 0$ implies $\wt{p}_i-\wt{p}_j\geq 0$ for all pairs $i,j\in[M]$. For the equality case, i.e., $p_i=p_j$, the error incurred in estimation does not affect $\ug$. Now, for distinct $(p_i,p_j)$ pairs let us define $\Delta:=\min_{i,j}\left\{\left|p_i-p_j\right|: p_i\neq p_j\right\}$. Then for $p_i-p_j\geq 0$, we have $p_i-p_j\geq \Delta$. Using this, we have
\begin{align*}
    \left(\wt{p}_i-\wt{p}_j\right)&= \left(\wt{p}_i-p_i\right)-\left(\wt{p}_j-p_j\right) +\left(p_i-p_j\right)\\
    &\geq -2\varepsilon +\Delta.
\end{align*}
Thus, taking $\varepsilon < \Delta/2$ (e.g. taking $\varepsilon=\Delta/3$), we can ensure that the error incurred in the estimation of probabilities $p_i$ does not affect the choice of global unitary. 

The proof of Lemma \ref{lem:sample-prob-with-error} follows directly by substituting the upper bound on $\varepsilon$ into the expression for $\wt{T}$. This shows that the estimation of probabilities $p_i$ for all $i\in[M]$ with additive error $\Delta/3$ and failure probability at most $\delta$ requires only $O\left(\log (M/\delta)/\Delta^2\right)$ samples of the unknown state.

\subsection{Robustness of observational ergotropy: Proof of Lemma \ref{lem:guess-ergobs-ergobs-diff}}
\label{sub:proofs-lems}
In this subsection, we formally prove Lemma \ref{lem:guess-ergobs-ergobs-diff}. In Sec.~\ref{sec:robustness-ergobs}, we divided the eigenvalues of $H$ into contiguous blocks $A_1, A_2, \cdots A_M$, such that the average energy of each block $\overline{E}_i$ is as stated in Eq.~\eqref{eq:block-avg-energy}. Then, observational ergotropy can be written as
\begin{align*}
    \ergobs{\rho}&=\Tr\left[H~\rho\right]-\Tr\left[\ug~\wt{\rho}~\ugdag\right]\\
    &=\Tr\left[H~\rho\right]-\sum_{j=1}^{M} p_j \overline{E}_j,
\end{align*}
where $\wt{\rho}$ is the coarse grained state, and in the last line we have used the fact that
$$
\Tr\left[\ug~\wt{\rho}~\ugdag\right]=\sum_{j=1}^{M} p_j \overline{E}_j,
$$
following Eq.~\eqref{eq:block-avg-energy}.

Now, as explained in Sec.~\ref{sec:robustness-ergobs}, for estimating \emph{guessed observational ergotropy}, the experimentalist first sorts the estimated probabilities $\wt{p}_{i}$, using the permutation $\pi$. Now, the global unitary $\wt{U}_{\mathrm{gl}}$ maps the eigenvectors of $P_{\pi(i)}$ to the energy block $A_i$. So, if $p_{\pi(i)}$ is the \emph{true} probability of the block $P_{\pi(i)}$ which the experimentalist belives to be the $i$-th largest, then
\begin{align*}
\wt{\mathrm{Erg}}_{\mathrm{obs}}(\rho)&=\Tr[H~\rho]-\Tr[\wt{U}_{\mathrm{gl}}~\wt{\rho}~\wt{U}^{\dag}_{\mathrm{gl}}]\\
&=\Tr[H~\rho] - \sum_{i=1}^{M} p_{\pi(i)}\overline{E}_i.
\end{align*}
Therefore, we obtain
\begin{align*}
    \left|\wt{\mathrm{Erg}}_{\mathrm{obs}}(\rho)-\ergobs{\rho}\right|&= \left|\Tr\left[\ug~\wt{\rho}~\ugdag\right]-\Tr[\wt{U}_{\mathrm{gl}}~\wt{\rho}~\wt{U}^{\dag}_{\mathrm{gl}}]\right|\\ &=\left|\sum_{j=1}^M \left(p_j-p_{\pi(j)}\right)\overline{E}_j\right|\\
    &\leq \|H\|\sum_{j=1}^M\left|p_j-p_{\pi(j)}\right|,
\end{align*}
where the upper bound in the last line follows from the triangle inequality.

\subsection{Randomized quantum algorithm for work extraction from local Hamiltonians}
\label{subsec:pauli-measurement}
For this section, consider any $n$-qubit Hamiltonian
$$
H=\sum_{j=1}^{K}h_j \Lambda_j,
$$
where each $\Lambda_j$ is a string of $n$-qubit Pauli operators and $h=\sum_{j}|h_j|$. Without loss of generality, we assume that $h_j\in \mathbb{R}^{+}$. This is because, in general, these coefficients can be written as $h_j = e^{i\phi}c_j$, where $c_j\in \mathbb{R}^{+}$. In this case, the imaginary phase can be absorbed into the description of the Pauli operator itself, i.e., $\Lambda_j \mapsto e^{i\phi}\Lambda_j$.  

In this case, we can replace the POVM measurements of $H$ in Protocol \ref{prot1} with the measurements of randomly sampled $\Lambda_j$ which takes values in $\{-1,1\}$. These modifications simplify the measurement scheme while still ensuring that the randomized algorithm estimates observational ergotropy $\ergobs{\rho}$ to $\varepsilon$-additive accuracy, with a success probability of at least $1-\delta$. We outline the steps via Protocol \ref{prot1mod}.

\RestyleAlgo{boxruled}
\begin{protocol}
\caption{$\texttt{Work Extraction from local Hamiltonians}$}
  \label{prot1mod}
\begin{enumerate}
    \item With probability $1/2$, sample $\Lambda\sim \{h_i/h, \Lambda_i\}_{i=1}^M$ and measure $\Lambda$ on the state $\rho$.
    \item With probability $1/2$, execute steps $(a)-(c)$.
    \begin{enumerate}
        \item Measure $\rho$ in the basis spanned by the projectors $\{P_i\}_{i=1}^{M}$ and ignore the outcomes to obtain the post-measurement state $\bar{\rho}$.
        \item Sample a unitary $U_r$ uniformly at random from the $r$-qubit Pauli group, where $r=\log_2(N/M)$, and apply the unitary $\ug(\mathbb{I}_M\otimes U_r)$ to $\bar{\rho}$ to obtain 
        \begin{align*}
        \rho'=\ug(I_M\otimes U_r) \bar{\rho} (I_M\otimes U_r)^{\dag}\ugdag
        \end{align*}
        \item Sample $\Lambda\sim \{h_i/h, \Lambda_i\}_{i=1}^M$ and measure $-\Lambda$ in the state obtained after step (b) above.
    \end{enumerate}   
  \item Store the outcome of the measurement of the $j^{\mathrm{th}}$ iteration into $\mu_j$.
  \item Repeat Steps 1 and 2(a)-(c), a total of $T$ times. 
  \item Output $\mu=\dfrac{2h}{T}\sum_{j=1}^{T} \mu_j$.
\end{enumerate}
\end{protocol}

In particular, Step 1 and Step 2(c) of Protocol \ref{prot1} is modified as follows: instead of measuring $H$ (or $-H$), sample a Pauli string $\Lambda$ according to the ensemble $\{h_i/h, \Lambda_i\}_{i=1}^{M}$, and perform a POVM measurement of $\Lambda$ with respect to $\rho$ in Step 1 (or equivalently $-\Lambda$ in Step 2(c)). Note that $\mathbb{E}[\Lambda]=H/h$, and the $j^{\mathrm{th}}$ iteration of Steps 1, and Steps 2(a)-(c) in Protocol \ref{prot1mod} results in a random variable $\mu_j$ in Step 3 such that
$$
2~\mathbb{E}[\mu_j]=\dfrac{1}{h}\Tr\left[H\left(\rho-\ug~\wt{\rho}~\ug^{\dag}\right)\right].
$$
Then, the overall output of Protocol \ref{prot1mod} (Step 5) is 
$$
\mu=\dfrac{2h}{T}\sum_{j=1}^{T} \mu_j.
$$
As before, Hoeffding's inequality ensures that
$$
\mathrm{Pr}\left[\left|\mu -\ergobs{\rho} \right| \geq \varepsilon \right]\leq 2\exp\left[-T h^{-2}\varepsilon^2 /8\right].
$$
So, we have
$$
\left|\mu -\ergobs{\rho} \right| \leq \varepsilon,
$$
with probability $1-\delta$ as long as
$$
T=O\left(\dfrac{h^2\log(1/\delta)}{\varepsilon^2}\right).
$$

\section{Discussion and Future Work}\label{sec:discussion}
We consider the problem of extracting work from a physical system described by a known Hamiltonian $H$, but prepared in an unknown quantum state $\rho$, accessible only through copies provided in a black-box manner. In this paper, we determine the number of samples required to estimate the extractable work in such a scenario. We rigorously show that extracting any non-negligible amount of work from a single copy of $\rho$ is essentially impossible in the appropriate limit (see Theorem~\ref{thm:impossibility-work}). This underscores a fundamental limitation of single-shot thermodynamic protocols when the state is not known and cannot be reconstructed with sufficient precision from just one instance.

When multiple copies are available, however, we demonstrate that only $O(1/\varepsilon^2)$ samples suffice to estimate the amount of extractable work, as quantified by \emph{observational ergotropy}, with high probability and $\varepsilon$-additive accuracy. Notably, the same number of copies is required not only to estimate the ergotropy but also to meaningfully leverage the state as an energy resource---for instance, for charging a quantum battery~\cite{Andolina2018} or implementing arbitrary quantum gates in an energy-aware and efficient manner~\cite{Chiribella2021}.

Observational ergotropy depends intricately on the choice of measurement and, in certain cases, may even be non-positive. While it is theoretically possible to maximize this quantity over all measurement settings to ensure positivity, doing so is practically infeasible due to constraints on the number of state copies and circuit complexity. Consequently, we focus on the practically relevant setting where only a single quantum measurement is allowed. How the sample complexity and accuracy of ergotropy estimation would be affected if multiple or adaptive measurement settings were permitted is left open in our work. Exploring such scenarios could yield improved protocols for estimating extractable work or for identifying near-optimal measurement bases with minimal overhead.

To achieve sample efficiency, our quantum algorithm employs a subtle combination of randomized sampling and a single quantum measurement, ensuring that the entire procedure remains \emph{qubit-efficient}. This feature is especially important given the current limitations of quantum hardware. This is non-trivial as other ways of estimating observational ergotropy can be inefficient. In fact, in Sec.~\ref{append-algo-single ancilla} of the Appendix, we illustrate how more elaborate primitives, such as the Linear Combination of Unitaries (LCU) \cite{childs2012hamiltonian, Chakraborty2024implementingany}, can be integrated into our algorithmic framework to estimate extractable work. While such methods provide an alternative route to estimating extractable work, they introduce significant overhead by increasing the sample complexity exponentially and requiring additional ancilla qubits. This highlights a key trade-off between algorithmic generality and physical resource efficiency.

Beyond efficiency, we also provide rigorous bounds that quantify the robustness of observational ergotropy estimation in the presence of realistic errors. In particular, we show how imperfections in measurement statistics or deviations in average energy values translate into variations in the estimation. These bounds not only certify the reliability of our protocol under errors but also clarify the regimes where observational ergotropy provides a faithful proxy for true extractable work.

Our work thus paves the way for using quantum computers to estimate thermodynamically meaningful quantities of unknown quantum systems in a manner that is both resource-aware and experimentally viable. Going forward, several promising research directions emerge from our study. One natural extension is to explore adaptive or multi-basis measurement strategies, where outcomes from previous measurements guide the choice of future ones. Such feedback-based approaches could potentially enhance estimation accuracy or reduce sample requirements. Moreover, while observational ergotropy provides a physically operational quantity, it would be valuable to extend similar estimation techniques to other thermodynamic primitives such as free energy differences, non-equilibrium work distributions, or coherence-based contributions to extractable work. These quantities are relevant in broader thermodynamic and resource-theoretic frameworks, and their estimation in black-box settings remains largely unexplored.

Another promising direction is to develop protocols that integrate ergotropy estimation into broader thermodynamic cycles, such as those governing the operation of quantum batteries, engines, or refrigerators. In such contexts, real-time estimation of work potential could serve as a feedback mechanism to control device operation, thereby enhancing energy efficiency.

\medskip
\noindent
\subsection*{Acknowledgments}
We thank Karen V. Hovhannisyan, Rick P. A. Simon, Janet Anders, Hamed Mohammady, and Riccardo Castellano for helpful discussions and comments. The authors acknowledge support from the Ministry of Electronics and Information Technology (MeitY), Government of India, under Grant No. 4(3)/2024-ITEA. AG and SD acknowledge support from the Science and Engineering Research Board, Department of Science and Technology (SERB-DST), Government of India, under Grant No.~SRG/2023/000217. SC and SH acknowledge funding from the Science and Engineering Research Board, Department of Science and Technology (SERB-DST), Government of India under Grant No.~SRG/2022/000354.~SC also acknowledges funding from Fujitsu Ltd, Japan. SC, SD, and US thank IIIT Hyderabad for support via the Faculty Seed Grant.

\subsection*{Data availability}
The python codes used in this study are publicly available in the GitHub repository \cite{git}.

\onecolumngrid
\newpage
\appendix
\renewcommand{\thelemma}{A\arabic{lemma}}
\section*{Appendix}
\label{sec:supp-mat}

In the Appendix, we provide additional examples that (a) illustrate the sample complexity of our work extraction algorithm (see Appendix~\ref{append-examples1} and Appendix~\ref{append-examples2}), and (b) shows the tightness of the bound we obtained for the robustness of observational ergotropy (Appendix~\ref{append:ex-rob}). We also derive a non-trivial bound on the deviation between true ergotropy and observational ergotropy in Appendix~\ref{append-deviation-erg-ergobs}. Finally, in Appendix~\ref{append-algo-single ancilla}, an ancilla-assisted work extraction protocol is also discussed, which, although interesting, results in a worse sample complexity than our protocol in the main text.

\section{Example of a \texorpdfstring{$3$}{}-qubit system }
\label{append-examples1}
Let us consider a $3$-qubit system with the Hamiltonian
\begin{align}
\label{eq:toy-ham}
H^{(8)}=H \otimes \mathbb I_2 \otimes \mathbb I_2  + \mathbb I_2 \otimes H \otimes \mathbb I_2 + \mathbb I_2 \otimes\mathbb I_2 \otimes H,
\end{align}
where $H = \frac{1}{2} (\mathbb I_2 - Z)$. Here $n=3$ and $N=8$. The spectral decomposition of this Hamiltonian is given by $H^{(8)}=\sum_{a=1}^4 E_a \sum_{\kappa=1}^{g(E_a)}\ket{E_a;\kappa}$, where 
   \begin{align*}
       &E_1=0, g(E_1)=1, \ket{E_1;1}=\ket{000}:=\ket{\phi_1};\\
       &E_2=1, g(E_2)=3, \ket{E_2;1}=\ket{001}:=\ket{\phi_2}, \ket{E_2;2}=\ket{010}:=\ket{\phi_3}, \ket{E_2;3}=\ket{100}:=\ket{\phi_4};\\
       &E_3=2, g(E_3)=3, \ket{E_3;1}=\ket{011}:=\ket{\phi_5}, \ket{E_3;2}=\ket{101}:=\ket{\phi_6}, \ket{E_3;3}=\ket{110}:=\ket{\phi_7};\\
       &E_4=3, g(E_4)=1, \ket{E_4;1}=\ket{111}:=\ket{\phi_8}.
   \end{align*}
Further let $Q_a=\sum_{\kappa=1}^{g(E_a)}\ket{E_a;\kappa}\bra{E_a;\kappa}$ with $a=1,\cdots,4$. Consider the unknown state to be  $\rho=(1/3)\sum_{i=1}^3 \rho_i$, where 
    \begin{align}
    \label{eq:first-ex-rho}
       &\rho_1:=\frac{1}{3}\sum_{i,j=1}^3\ket{\phi_i}\bra{\phi_j};~\rho_2:=\frac{1}{3}\sum_{i,j=4}^6(-1)^{i+j}\ket{\phi_i}\bra{\phi_j};~\rho_3:=\frac{1}{2}\sum_{i=7}^8\ket{\phi_i}\bra{\phi_i}.
   \end{align}
   The energy $\Tr\left[H^{(8)}\rho\right]$ of state $\rho$ is given by $29/18$.
   A minimum energy state (such a state is not unique here) with respect to $\rho$ is given by
   \begin{align}
       \rho_p=\frac{1}{6}\left(2\sum_{i=1}^2\ket{\phi_i}\bra{\phi_i} +\sum_{i=3}^4 \ket{\phi_i}\bra{\phi_i} \right).
   \end{align}
   The energy $\Tr\left[H^{(8)}\rho_p\right]$ of state $\rho_p$ is given by $2/3$. The ergotropy of $\rho$ is given by
   \begin{align*}
    \erg{\rho}&=\Tr\left[H^{(8)}\left( \rho - \rho_p\right)\right]=\frac{29-12}{18}\approx 0.94.
\end{align*}

\subsection{Measurement in entangled basis}
\label{app:case1}
Let the projective measurement be given by $\{P_i\}_{i=1}^4$, where $P_i=\sum_{b=1}^2 \ket{\psi_i^{(b)}}\bra{\psi_i^{(b)}}$ for $i\in[4]$, where
\begin{align}
    &\ket{\psi_1^{(1)}}:=\frac{1}{\sqrt{2}}\left(\ket{\phi_1} + \ket{\phi_7}\right);
    \ket{\psi_1^{(2)}}:=\frac{1}{\sqrt{2}}\left(\ket{\phi_2} + \ket{\phi_8}\right);\\
    &\ket{\psi_2^{(1)}}:=\frac{1}{\sqrt{2}}\left(\ket{\phi_1} - \ket{\phi_7}\right);
    \ket{\psi_2^{(2)}}:=\frac{1}{\sqrt{2}}\left(\ket{\phi_2} - \ket{\phi_8}\right);\\
    &\ket{\psi_3^{(1)}}:=\frac{1}{\sqrt{2}}\left(\ket{\phi_3} + \ket{\phi_4}\right);
    \ket{\psi_3^{(2)}}:=\frac{1}{\sqrt{2}}\left(\ket{\phi_5} + \ket{\phi_6}\right);\\
    &\ket{\psi_4^{(1)}}:=\frac{1}{\sqrt{2}}\left(\ket{\phi_3} - \ket{\phi_4}\right);
    \ket{\psi_4^{(2)}}:=\frac{1}{\sqrt{2}}\left(\ket{\phi_5} - \ket{\phi_6}\right);
\end{align}
Indeed all $P_i$'s are two dimensional projectors and in fact $P_1=\ket{\Phi_+}\bra{\Phi_+}\otimes \mathbb{I}_2$, $P_2=\ket{\Phi_-}\bra{\Phi_-}\otimes \mathbb{I}_2$, $P_3=\ket{\Psi_+}\bra{\Psi_+}\otimes \mathbb{I}_2$, and $P_4=\ket{\Psi_-}\bra{\Psi_-}\otimes \mathbb{I}_2$, where $\ket{\Phi_+}=\frac{1}{\sqrt{2}}\left(\ket{00}+\ket{11}\right)$, $\ket{\Phi_-}=\frac{1}{\sqrt{2}}\left(\ket{00}-\ket{11}\right)$, $\ket{\Psi_+} = \frac{1}{\sqrt{2}}\left(\ket{01}+\ket{10}\right)$, and $\ket{\Psi_-} = \frac{1}{\sqrt{2}}\left(\ket{01}-\ket{10}\right)$.

\begin{figure}
    \centering
    \includegraphics[width=\linewidth]{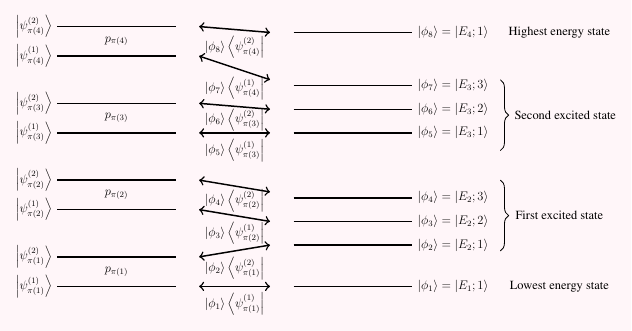}
    \caption{The schematic above shows the action of global unitary on measured state to yield a corresponding passive state for $n=3~ (N=8)$ and Hamiltonian $H^{(8)}$ given by Eq. \eqref{eq:toy-ham}. Let the projective measurement be given by $\{P_i\}_{i=1}^4$, where $P_i=\sum_{b=1}^2 \ket{\psi_i^{(b)}}\bra{\psi_i^{(b)}}$ and $i=1,\cdots,4$. Let $\rho$ be the unknown state and $p_i=\Tr\left[P_i\rho\right]$. Let $\pi$ be the permutation on $4$ objects such that $p_{\pi(i)}\geq p_{\pi(i+1)}$ for all $i\in[3]$. Then we define $\ug= \ket{\phi_1}\bra{\psi_{\pi(1)}^{(1)}} + \ket{\phi_2}\bra{\psi_{\pi(1)}^{(2)}} + \ket{\phi_3}\bra{\psi_{\pi(2)}^{(1)}}+\ket{\phi_4}\bra{\psi_{\pi(2)}^{(2)}} + \ket{\phi_5}\bra{\psi_{\pi(3)}^{(1)}}+ \ket{\phi_6}\bra{\psi_{\pi(3)}^{(2)}} + \ket{\phi_7}\bra{\psi_{\pi(4)}^{(1)}} + \ket{\phi_8}\bra{\psi_{\pi(4)}^{(2)}}$.}
    \label{fig:global-unitary}
\end{figure}

Let us now compute the probabilities $\Tr[P_i\rho]$. We have $\Tr[P_1\rho] =\frac{5}{18}$, $\Tr[P_2\rho] = \frac{5}{18}$, $\Tr[P_3\rho] =\frac{1}{9}$, and $\Tr[P_4\rho] = \frac{1}{3}$. The average state after the measurement is given by $\widetilde{\rho}=\sum_{i=1}^4P_i\rho P_i$. The energy $\Tr\left[H^{(8)} \wt{\rho}\right]$ is given by $3/2$. The eigenvalues of $\widetilde{\rho}$ are given by $\frac{1}{18} \left(\sqrt{5}+3\right),\frac{7}{36},\frac{7}{36},\frac{1}{9},\frac{1}{12},\frac{1}{12},\frac{1}{18} \left(3-\sqrt{5}\right),0$. The minimum energy state corresponding to this state is given by
\begin{align}
   \widetilde{\rho}_p&=\frac{1}{18} \left(\sqrt{5}+3\right) \ket{\phi_1}\bra{\phi_1}+ \frac{7}{36} \sum_{i=2}^3\ket{\phi_i}\bra{\phi_i} + \frac{1}{9} \ket{\phi_4}\bra{\phi_4} + \frac{1}{12} \sum_{i=5}^6\ket{\phi_i}\bra{\phi_i}+ \frac{1}{18} \left(3-\sqrt{5}\right) \ket{\phi_7}\bra{\phi_7} .
\end{align}
The ergotropy of the state $\widetilde{\rho}_p$ is given by
\begin{align*}
    \erg{\widetilde{\rho}_p}&=\Tr\left[H^{(8)}\left( \wt{\rho} - \wt{\rho}_p\right)\right]=\frac{3}{2}-\frac{1}{18} \left(21-2\sqrt{5}\right)\approx 0.58.
\end{align*}

After averaging $\wt{\rho}$ with respect to Haar measure, we get $\rho'=\frac{1}{36}\left(5 P_1+ 5P_2+2 P_3+6P_4\right)$. The energy of the state $\rho'$ is given by $3/2$. The eigenvalues (in decreasing  order) $\{\lambda_i\}_{i=1}^8$ and the corresponding set of eigenvectors $\{\ket{\lambda_i}\}_{i=1}^8$ of $\rho'$ are given by
\begin{align*}
  &\lambda_1 =  \frac{1}{6},~\ket{\lambda_1}=\ket{\psi_4^{(1)}};~\lambda_2 =  \frac{1}{6},~\ket{\lambda_2}=\ket{\psi_4^{(2)}};~\lambda_3 =  \frac{5}{36},~\ket{\lambda_3}=\ket{\psi_1^{(1)}};~\lambda_4 =  \frac{5}{36},~\ket{\lambda_4}=\ket{\psi_1^{(2)}};\\
  &\lambda_5 =  \frac{5}{36},~\ket{\lambda_5}=\ket{\psi_2^{(1)}},~\lambda_6 =  \frac{5}{36},~\ket{\lambda_6}=\ket{\psi_2^{(2)}};~\lambda_7 =  \frac{1}{18},~\ket{\lambda_7}=\ket{\psi_3^{(1)}};~\lambda_8 =  \frac{1}{18},~\ket{\lambda_8}=\ket{\psi_3^{(2)}}.
\end{align*}
Let us define permutation $\pi$ as $(1~ 4 ~3 ~2)$, i.e., $\pi(1)=4,~\pi(4)=3,~\pi(3)=2, ~\pi(2)=1$.  Using this, let us consider a global unitary (see Fig.~\ref{fig:global-unitary})
\begin{align*}
    \ug&=\ket{\phi_1}\bra{\psi_{\pi(1)}^{(1)}} + \ket{\phi_2}\bra{\psi_{\pi(1)}^{(2)}} + \ket{\phi_3}\bra{\psi_{\pi(2)}^{(1)}} + \ket{\phi_4}\bra{\psi_{\pi(2)}^{(2)}} + \ket{\phi_5}\bra{\psi_{\pi(3)}^{(1)}}+\ket{\phi_6}\bra{\psi_{\pi(3)}^{(i)}}+ \ket{\phi_7}\bra{\psi_{\pi(4)}^{(1)}}+\ket{\phi_8}\bra{\psi_{\pi(4)}^{(2)}}\\
    &=\ket{\phi_1}\bra{\psi_{4}^{(1)}} + \ket{\phi_2}\bra{\psi_{4}^{(2)}} + \ket{\phi_3}\bra{\psi_{1}^{(1)}} + \ket{\phi_4}\bra{\psi_{1}^{(2)}} + \ket{\phi_5}\bra{\psi_{2}^{(1)}}+\ket{\phi_6}\bra{\psi_{2}^{(i)}}+ \ket{\phi_7}\bra{\psi_{3}^{(1)}}+\ket{\phi_8}\bra{\psi_{3}^{(2)}}.
\end{align*}
 The minimum energy state for $\rho'=\ug\rho'\ugdag$ is then given by
\begin{align}
   \rho'_p&=\frac{1}{36}\left[6\sum_{i=1}^2 \ket{\phi_i}\bra{\phi_i} + 5 \sum_{i=3}^6 \ket{\phi_i}\bra{\phi_i} +2 \sum_{i=7}^8 \ket{\phi_i}\bra{\phi_i}\right].
\end{align}
The ergotropy of the state $\rho'$ is given by $\erg{\rho'}=\Tr\left[H^{(8)}\left( \rho' - \rho'_p\right)\right]=\frac{2}{9}\approx 0.22$.
while the observational ergotropy of the state $\rho$ is given by $\ergobs{\rho}=\Tr\left[H^{(8)}\left( \rho - \rho'_p\right)\right]=\frac{1}{3}\approx 0.33$ (See Table \ref{tab:ex1} and Fig.~\ref{fg.samp2}).

\begin{table}[]
\centering
\begin{tabular}{|c|c|c|c|c|}
\hline
\multirow{2}{*}{State} & \multirow{2}{*}{Energy} & \multirow{2}{*}{Ergotropy} &  \multicolumn{2}{c|}{Observational Ergotropy}\\
 \cline{4-5} &  &  & Entangled basis (Sec.~\ref{app:case1}) & Product basis (Sec.~\ref{app:case2}) \\
\hline
 $\rho$ & $\approx 1.61$ & $\approx 0.94$ & $\approx 0.33$ & $\approx 0.22$ \\
\hline
\end{tabular}
\caption{Ergotropy and observational ergotropy. For state $\rho=(1/3)\sum_{i=1}^3\rho_i$, where $\rho_1$, $\rho_2$, and $\rho_3$ are given in Eq. \eqref{eq:first-ex-rho}, the table shows values of ergotropy and observational ergotropy for two different measurement settings.}
\label{tab:ex1}
\end{table}

\begin{figure}[ht!]
\centering
\subfloat[Measurement in entangled basis]{
{\includegraphics[width=0.45\textwidth]{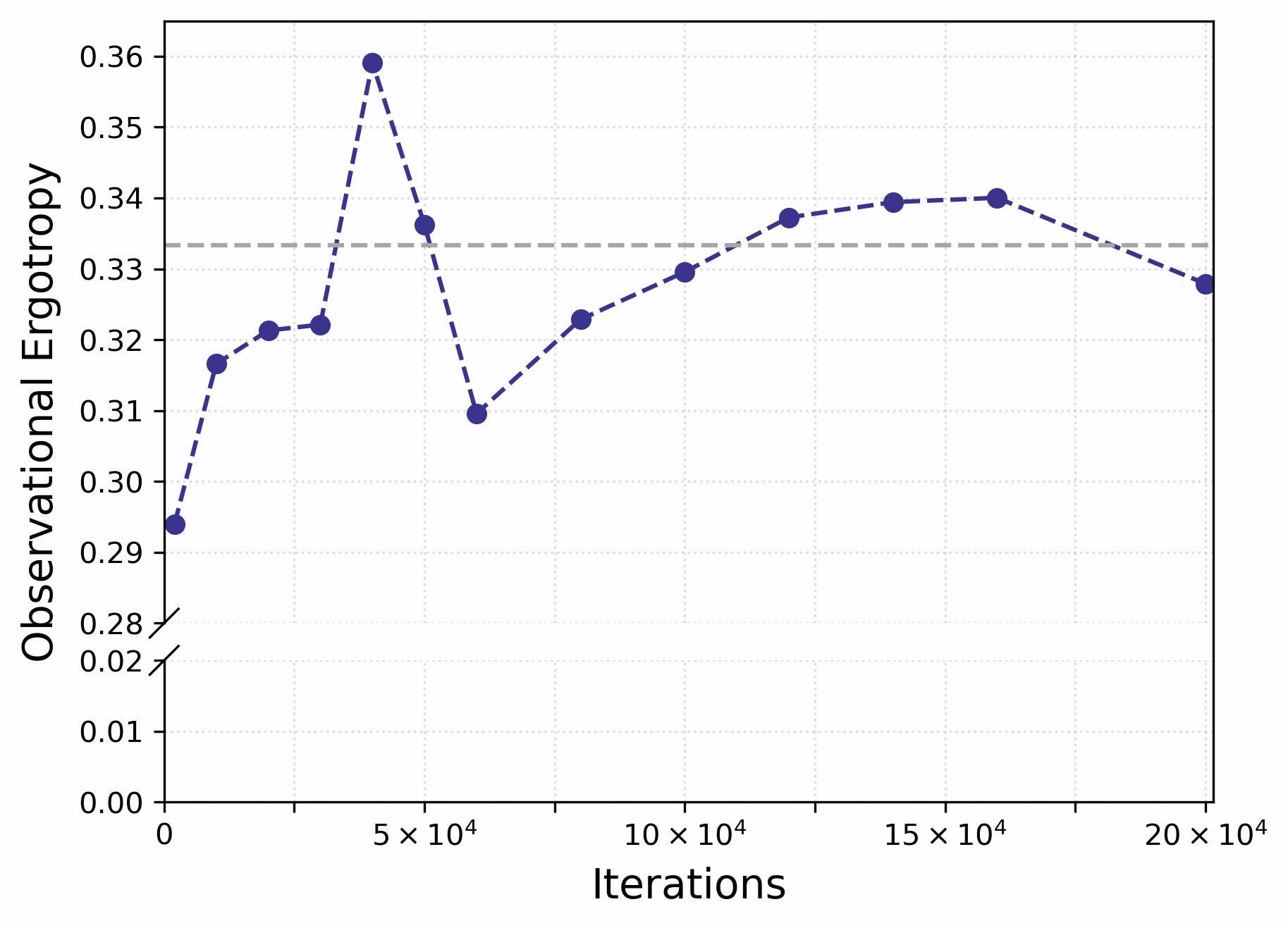}\label{fg.samp2}}
}
\subfloat[Measurement in product basis]{
{\includegraphics[width=0.45\textwidth]{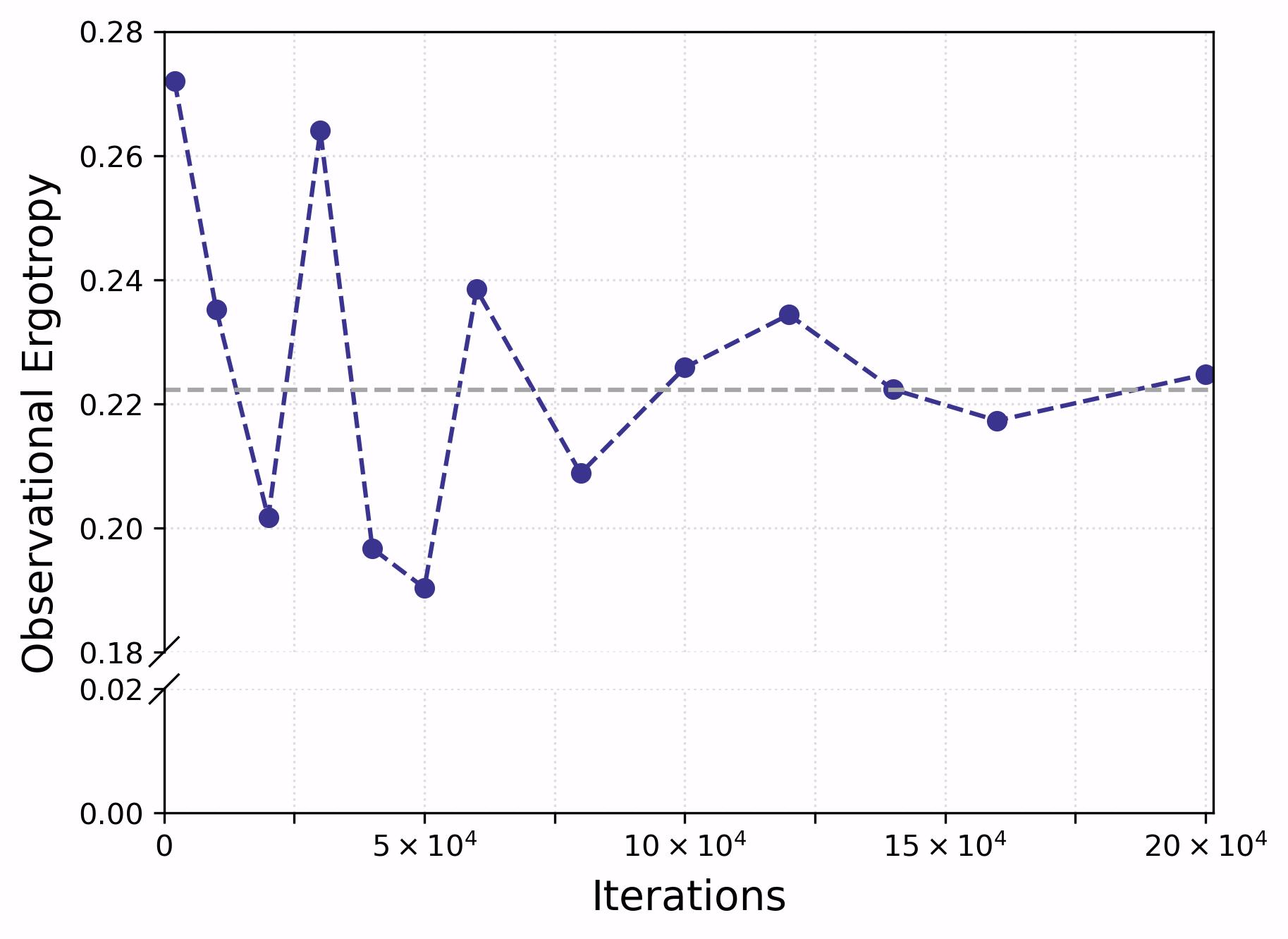}\label{fg.samp1}}
}
\caption{Observational ergotropy and sample complexity. Consider a state $\rho=(1/3)\sum_{i=1}^3\rho_i$, where $\rho_1$, $\rho_2$, and $\rho_3$ are given in Eq. \eqref{eq:first-ex-rho}. Figs. (\ref{fg.samp2}) and (\ref{fg.samp1}) show the convergence of our protocol from main text to the exact value of observational ergotropy, as we increase number of samples, with respect to measurements in entangled (Sec.~\ref{app:case1}) and product (Sec.~\ref{app:case2}) bases, respectively.}
\label{fg.samp}
\end{figure}

\subsection{Measurement in product basis}
\label{app:case2}
Let the projective measurement be given by $\{K_i\}_{i=1}^4$, where $K_1=\ket{00}\bra{00}\otimes \mathbb{I}_2$, $K_2=\ket{01}\bra{01}\otimes \mathbb{I}_2$, $K_3=\ket{10}\bra{10}\otimes \mathbb{I}_2$, and $K_4= \ket{11}\bra{11}\otimes \mathbb{I}_2$.
%
We have $\Tr[K_1\rho] = 2/9$, $\Tr[K_2\rho] = 2/9$, $\Tr[K_3\rho] = 2/9$, and $\Tr[K_4\rho] = 1/3$. The average state after the measurement is given by $\widetilde{\sigma}=\sum_{i=1}^4K_i\rho K_i$. The energy $\Tr\left[H^{(8)} \wt{\sigma}\right]$ is given by $29/18$. The eigenvalues of $\widetilde{\sigma}$ are given by $\frac{2}{9},\frac{2}{9},\frac{1}{6},\frac{1}{6},\frac{1}{9},\frac{1}{9},0,0$. The minimum energy state corresponding to this state is given by
\begin{align}
   \widetilde{\sigma}_p&=\frac{2}{9}\sum_{i=1}^2 \ket{\phi_i}\bra{\phi_i}+ \frac{1}{6}\sum_{i=3}^4\ket{\phi_i}\bra{\phi_i} + \frac{1}{9} \sum_{i=5}^6\ket{\phi_i}\bra{\phi_i}.
\end{align}
The ergotropy of the state $\wt{\sigma}$ is given by $\erg{\wt{\sigma}}=\Tr\left[H^{(8)}\left( \wt{\sigma} - \wt{\sigma}_p\right)\right]=\frac{29}{18}-1\approx 0.61$. After averaging $\wt{\sigma}$ wrt to Haar measure, we get
\begin{align}
  \sigma'&=\frac{1}{18}\left(2K_1+ 2K_2+2K_3+3K_4\right).
\end{align}

The energy of the state $\sigma'$ is given by $29/18$. The minimum energy state for $\sigma'$ is given by
\begin{align}
   \sigma'_p&=\frac{1}{6}\sum_{i=1}^2 \ket{\phi_i}\bra{\phi_i} +\frac{1}{9}  \sum_{i=3}^8 \ket{\phi_i}\bra{\phi_i} 
\end{align}
The energy of the state $\sigma'_p$ is given by $\Tr\left[H^{(8)} \sigma'_p\right]= \frac{25}{18}$. The ergotropy of the state $\sigma'$ is given by $\erg{\sigma'}=\Tr\left[H^{(8)}\left( \sigma' - \sigma'_p\right)\right] =\frac{4}{18}\approx 0.22$. The observational ergotropy of the state $\rho$ is given by $\ergobs{\rho}=\Tr\left[H^{(8)}\left( \rho - \sigma'_p\right)\right]=\frac{2}{9}\approx 0.22$ (See Table \ref{tab:ex1} and Fig.~\ref{fg.samp1}).

\section{Example of \texorpdfstring{$4$}{} qubit translationally invariant spin chain}
\label{append-examples2}
In this section we consider translationally invariant spin chain, in particular, one dimensional Heisenberg XXX model, as an example. For $n$ spins such a spin chain is described by the Hamiltonian
\begin{align}
    H_{n} = -J \sum_{j=1}^n\sigma_z^{(j)} \otimes \sigma_z^{(j+1)} - h \sum_{j=1}^n\sigma_x^{(j)},
\end{align}
where $J$ is the spin coupling constant, $h$ is the transverse magnetic field and the periodic boundary condition is ensured by identification $\sigma_z^{(n+1)} = \sigma_z^{(1)}$.
\begin{figure}[ht!]
\centering
\includegraphics[width=0.45\textwidth]{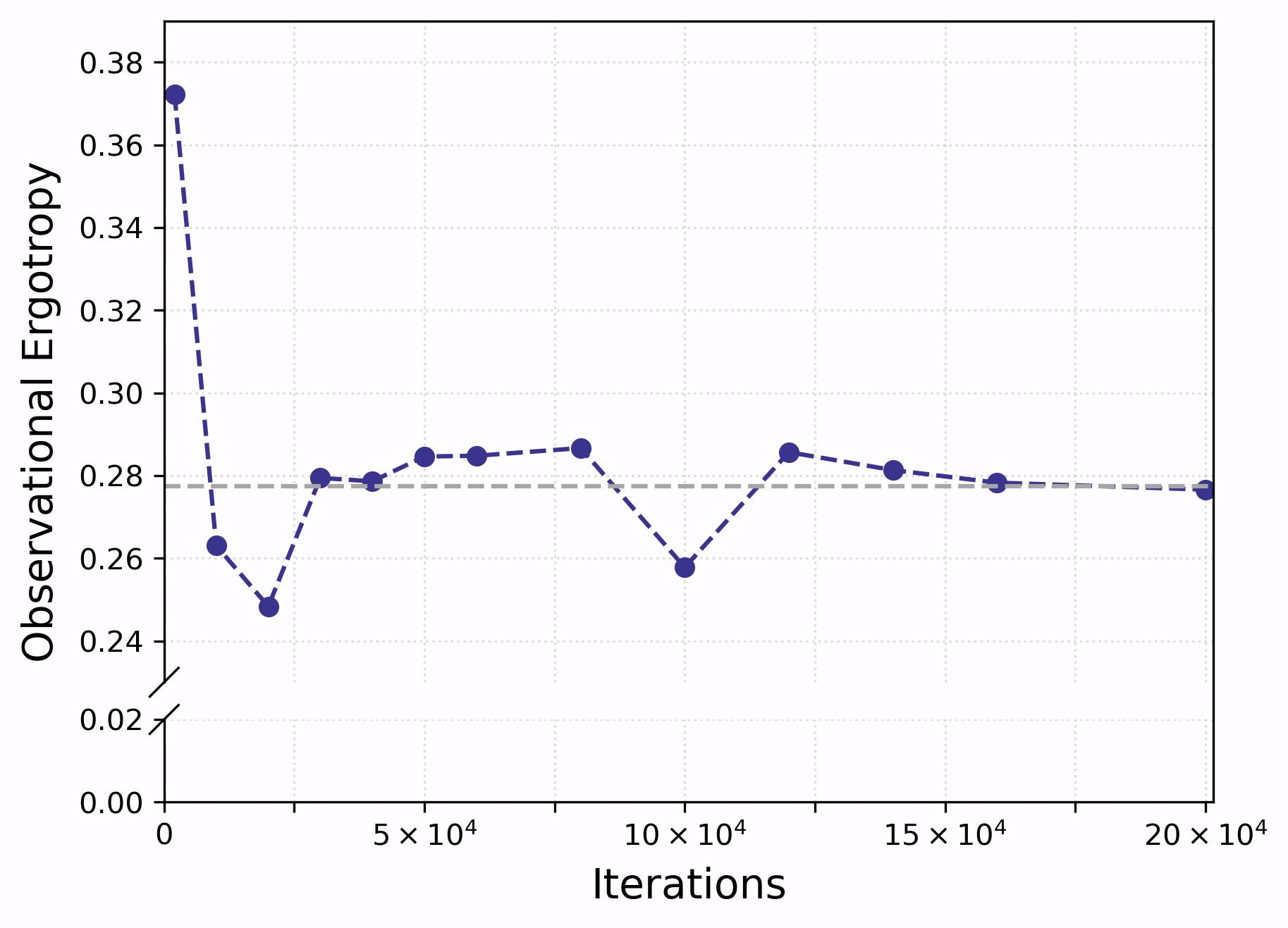}
\caption{$n=4$ Heisenberg XXX model. For the state in Eq. \eqref{appeq:4-state}, the exact observational ergotropy of $\eta_4$ is equal to $\approx 0.28$. The figure shows how our protocol for work extraction reaches to $\approx 0.28$ as we increase the number of samples in accordance with the bound on the sample complexity.}
\label{fg.samp2-spin}
\end{figure}
The sign of $J$ ensures whether the ground state is ferromagnetic or antiferromagnetic \cite{Anderson1959, Auerbach1994, Kittel2004}. Here we take $J=1$ and $h=1$. For the case of $4$ spins, we consider the state 
\begin{align}
\label{appeq:4-state}
  \eta_4=\frac{1}{4}\ket{\phi_1}\bra{\phi_1} + \frac{1}{4}\ket{\phi_2}\bra{\phi_2} + \frac{1}{2}\ket{\phi_3}\bra{\phi_3},
\end{align}
where $\ket{\phi_1}=\ket{0010}$, $\ket{\phi_2}=\ket{0111}$ and 
\begin{align}
    \ket{\phi_3} = \frac{1}{2}\left(\ket{0000} + \ket{0001} + \ket{1110}+ \ket{1111}\right).
\end{align}
The ergotropy of $\eta_4$ is computed to be  $\approx 2.86$. Further, to obtain observational ergotropy, we perform the projective measurement defined by the projectors $L_i$ for $i=[4]$, where 
 \begin{align}
       L_1=\ket{00}\bra{00}\otimes \mathbb{I}_4;~L_2=\ket{01}\bra{01}\otimes \mathbb{I}_4;~
       L_3=\ket{10}\bra{10}\otimes \mathbb{I}_4;~
       L_4=\ket{11}\bra{11}\otimes \mathbb{I}_4.\label{eq:proj4-4}
   \end{align}
With the above projectors, the observational ergotropy of $\eta_4$ is given by $\approx 0.28$. In Fig.~\ref{fg.samp2-spin}, we show how our protocol reaches this value as we increase the number of samples.

\section{Deviation between true ergotropy and observational ergotropy}
\label{append-deviation-erg-ergobs}
Estimating the ergotropy of an unknown quantum state directly is notoriously resource-intensive, since it requires complete spectral information about the state. Our protocol in Sec.~\ref{subsec:work-multiple-copy} (Protocol \ref{prot1}), on the other hand, circumvents this by efficiently estimating the \emph{observational ergotropy}, which depends only on coarse-grained information obtained from a small number of samples. From a thermodynamic perspective, it is natural to expect that the true ergotropy of a state $\rho$, the maximum work extractable under optimal global control, should always exceed its observational counterpart, which is limited by the coarse description accessible in practice. Put differently, while ergotropy leverages the full microscopic structure of $\rho$, observational ergotropy reflects the work one can certify from only a few samples of the state. In this sub-section, we sharpen this intuition by establishing a non-trivial lower bound on the difference between the two quantities.

Recall that to estimate observational ergotropy, we introduced a collection of $M$ orthogonal projectors $\{P_i\}_{i=1}^{M}$, each of dimension $d=2^r$, satisfying $\sum_{i} P_i=\mathbb{I}_N$. These projectors define the probabilities $p_i=\Tr[P_i\rho]$, and specify the coarse-grained state $\wt{\rho}$. Let $E_1\leq E_2\leq \cdots E_{N}$ denote the eigenvalues of the system Hamiltonian $H$, listed in non-decreasing order, and counted with degeneracies. Then we have the following lemma.
\begin{lemma}
\label{lemma:lower-bound}
Let $N=2^n$, and $H$ be the total Hamiltonian of an $n$-qubit system in the state $\rho$, with eigenvalues $E_1\leq E_2\leq\cdots\leq E_N$. Define the set of orthogonal projectors $\{P_i\}_{i=1}^{M}$, each of dimension $d=2^r$, where $M=N/d$, satisfying $\sum_{i}P_i=\mathbb{I}_N$. Furthermore, let $p_i=\Tr[P_i\rho]$, and $p_{\max}=\max_i p_i$. Then,
\begin{align}
   \erg{\rho}-\ergobs{\rho} \ge \max\left\{0, \left(\frac{p_{\max}\cdot S}{d} - \frac{\Tr[H]}{N}\right)\right\},
\end{align}
where $S=\sum_{i=1}^{d}E_{i}$, is the sum of the $d$ lowest energy eigenvalues of $H$.
\end{lemma}

\begin{proof}
 We first show that 
 $$
 \erg{\rho}-\ergobs{\rho} \ge 0.
 $$ 
 Note that for non-negative weights $q_j\geq 0$, and real valued functions $f_j(x)$, if $\min_x \sum_{j} q_j f_j(x)= \sum_{j} q_j f_j(x^*)$, then since $f_j(x^*) \geq \min_x f_j(x)$, we have $$\min_x \sum_{j} q_j f_j(x)\geq  \sum_{j} q_j \min_xf_j(x).
 $$ 
 Applying the integral form of this inequality to the expression for observational ergotropy, we obtain
\begin{align}
\ergobs{\rho}
&= \Tr\left[H~\rho\right] - \min_{V} \int_{U}d\mu(U)\,\Tr\left[H~V (U\rho U^{\dag})V^{\dag}\right] \nonumber \\
&\leq \Tr\left[H~\rho\right] - \int_{U}d\mu(U)\,\min_{V}\Tr\left[H~V (U\rho U^{\dag})V^{\dag}\right] \nonumber \\
&= \Tr\left[H~\rho\right] - \min_{V}\Tr\left[H~V\rho V^{\dag}\right] = \erg{\rho}.
\end{align}
Thus, observational ergotropy of a state is at most the true ergotropy, i.e.\ $\erg{\rho}-\ergobs{\rho}\geq 0$.

For the second part, we prove that 
$$
 \erg{\rho}-\ergobs{\rho} \ge \dfrac{p_{\max}\cdot S}{d}-\dfrac{\Tr[H]}{N}.
$$
First observe that, without loss of generality, we can assume that the eigenvalues of $H$ are non-negative. This is because, for any $c>0$, we can always add $c\mathbb{I}_N$ to $H$,  without changing the difference between ergotropies. In other words, ergotropy differences are invariant under the shift $H\mapsto H+c\cdot \mathbb{I}_N$.

The coarse-grained state $\wt{\rho}$ is a block-diagonal state (see Eq.~\eqref{eq:randomized-state}) and has $M$ distinct eigenvalues $\{p_i/d\}_{i=1}^M$, each with degeneracy $d$. The maximum eigenvalue of $\wt{\rho}$ is given by $p_{\max}/d$, where $p_{\max}=\max_ip_i$. Let us reorder the block probabilities such that 
$$
p^{\downarrow}_1=p_{\max}\geq \cdots \geq p^{\downarrow}_M.
$$
Also, write the energy eigenvalues of $H$, denoted as $E_1\leq \cdots E_{2^n}$, in nondecreasing order. Define contagious block sums 
\begin{align}
    S_j=\sum_{i=(j-1)d+1}^{jd}E_i,  
\end{align}
for $j\in [1,M]$. Note that $S_1=S$. Then, the minimum over unitaries that permute the eigenbasis of $\wt{\rho}$ is obtained by pairing the largest eigenvalues of $\wt{\rho}$ with the smallest energy eigenvalues. Hence, 
\begin{align*}
    \min_V\Tr\left[H\, V\wt{\rho}V^\dag\right]&=\frac{1}{d}\sum_{j=1}^M p_{j}^{\downarrow} S_j\\
    &\geq \frac{p_{\max}S}{d},
\end{align*} 
where the lower bound in the last line follows from the fact that each term in the sum is non-negative. On the other hand, for any state $\rho$,
$$\min_V \Tr[H V \rho V^\dagger] \leq \mathbb{E}_V \Tr[H V \rho V^\dagger]=\Tr[H]/2^n.
$$  
Combining these, we conclude the proof of the lemma.
\end{proof}
This bound shows that the gap between the ergotropy and the observational ergotropy is controlled by two competing contributions: the probability weight $p_{\max}$ of the most likely projective measurement outcome, and the average energy contained in the lowest-energy subspace of dimension $d=2^r$. Whenever $p_{\max}$ is sufficiently larger than the uniform outcome $1/N$, the state $\rho$ necessarily places more weight in the low-energy sector than the coarse-grained description suggests, ensuring a strict positive lower bound on the gap. We emphasize, however, that Lemma~\ref{lemma:lower-bound} does not guarantee absolute positivity in all cases; when the block probabilities are too evenly distributed, the bound may vanish. Nonetheless, whenever it is positive, the lemma provides a concrete quantitative witness for the excess ergotropy accessible through optimized global unitaries.

\section{Estimating observational ergotropy using Linear Combination of Unitaries}
\label{append-algo-single ancilla}
In this section, we discuss a different approach for estimating observational ergotropy. It is, in principle, possible to use a different sampling strategy to implement a direct sum of Haar random matrices in Protocol \ref{prot1}. For this, we express this quantity (which is itself a unitary) as a Linear Combination of Unitaries (LCU) \cite{childs2012hamiltonian}. Observe that $V=\oplus_{i=1}^{M}\wt{V}_i$ can be expressed as as a linear combination of $2M$ unitaries as follows:
\begin{align*}
V=\oplus_{i=1}^{M} \wt{V}_i &= \sum_{i=1}^M \ketbra{i}{i}\otimes \wt{V}_i = \sum_{i=1}^{M} \frac{\mathbb{I}_M+\wt{I}_i}{2}\otimes \wt{V}_i= \sum_{j=1}^{2M} c_j W_j,
\end{align*}
where $\wt{I}_i=2\ketbra{i}{i} - \mathbb{I}_M$ is a  diagonal unitary matrix and $c_j=1/2$ for all $j\in[2M]$. The matrices $W_j$ above are unitaries and are defined as
\begin{align}
\label{eq:w-with-2M}
    W_j = \begin{cases}
       \mathbb{I}_M \otimes \wt{V}_j~~\mathrm{if}~1\leq j\leq M \vspace{0.2cm}\\
       \wt{I}_{j-M}\otimes \wt{V}_{j-M}~~\mathrm{if}~M < j\leq 2M.
    \end{cases}
\end{align}
Even though this framework has been central to the development of a number of near-optimal quantum algorithms \cite{childs2012hamiltonian, berry2015simulating, childs2017quantum, chakraborty2019power, vanApeldoorn2020quantumsdpsolvers}, it is quite resource-demanding as it requires several ancilla qubits and multi-qubit controlled operations to be implemented. However, recently, a randomized quantum algorithm for LCU was introduced \cite{chakraborty2019power}, which requires circuits of shorter depth and only a single ancilla qubit. So, instead of randomly sampling an element from the $r$-qubit Pauli group $\mathcal{P}(r)$ as described in Protocol \ref{prot1}, we can implement $V$ by using the Single-Ancilla LCU algorithm of Ref.~\cite{chakraborty2019power}. However, we prove that this leads to a substantial increase in the sample complexity.

Let us first consider a unitary $1$-design $Q:=\left\{R_1,\cdots, R_L\right\}$ of cardinality $L$ on the spaces $\mathcal{H}_i$ whose dimensions are the same for all $i\in[M]$. Here $R_i$, for all $i\in[L]$, are unitaries acting on $\mathcal{H}_j$ for all $j\in[M]$. Further, let us define a set $\mathsf{S}$ of cardinality $L^M$ as follows.
\begin{align}
\label{eq:direct-sum-1-design}
    \mathsf{S}:=\left\{\oplus_{i=1}^M\wt{V_i}~\vert~\wt{V}_i\sim Q,~\forall~i\in [M]~\right\}.
\end{align}
In the following lemma we prove that $\mathsf{S}$ is itself a unitary $1$-design on  $\oplus_{i=1}^M\mathcal{H}_i$.

\begin{lemma} The set $\mathsf{S}$, Eq. \eqref{eq:direct-sum-1-design}, of direct sums of unitaries $\wt{V}_i$ on Hilbert spaces  $\mathcal{H}_i$ for $i\in[M]$ sampled from a unitary $1$-design constitute a unitary $1$-design on $\oplus_{i=1}^M\mathcal{H}_i$ with respect to the product measure.
\end{lemma}
\begin{proof}
Note that averaging of an unknown state $\rho$ with respect to $\oplus_{i=1}^M\wt{V_i}=V$ sampled from the Haar measure $\mathrm{d}\mu\left(V\right)=\prod_{i=1}^M\mathrm{d}\mu\left( \wt{V}_i\right)$ produces 
\begin{align}
    \int_V V\rho V^{\dag} ~\prod_{i=1}^m \mathrm{d}\mu\left( \wt{V}_i\right) = \sum_i p_i \frac{P_i}{2^r},
\end{align}
where $p_i=\Tr\left[\rho P_i\right]$ and $\Tr\left[P_i\right]=2^r$ is the dimension of the each block on which $P_i$ is the projector. Now consider $\forall i \in \left[M \right]$, $\wt{V}_i$ is sampled from a unitary  $1$-design $Q =\left\{R_1,\cdots,R_L\right\}$ of cardinality $L$. Further, let the elements of $S$ be given by $V_i$ with $i=1,\cdots, L^M$. Each $V_i$ is of the form $\oplus_{i=1}^M\wt{V_i}$. To prove the lemma, it is sufficient to prove that 
\begin{align}
\label{eq:need-for-1-design}
 \int_V V \rho V^{\dagger} \prod_{i=1}^M\mathrm{d}\mu\left(\wt{V}_i\right)&=\frac{1}{L^M}\sum_{k=1}^{L^M}V_k \rho V_k^{\dagger}\nonumber\\
 &= \sum_{i,j=1}^M \frac{1}{L^M}\sum_{k=1}^{L^M}V_k P_i \rho P_j V_k^{\dagger},
\end{align}
where $P_i$ is a projector for all $i\in[M]$ and $\sum_{i=1}^M P_i=\mathbb{I}_N$. We analyze the cases $i=j ~\text{and}~ i\neq j$ in Eq. \eqref{eq:need-for-1-design}, separately.

\medskip
\noindent
{\bf {For the~$\mathbf{i=j}$} case}:~In this case, the action of $V_k=\left(\oplus_{i=1}^M \wt{V}_i\right)_k$ on $P_i\rho P_i$ corresponds to applying a unitary $R_l$ on the subspace defined by the projector $P_i$. For each choice of $R_l$ in a given subspace $P_i\rho P_i$ we can have $L^{M-1}$ possible choices of $R_l$ in remaining $M-1$ subspaces. This means that for each block we will have $\sum_{k=1}^{L^M} V_k P_i \rho P_i V_k^{\dagger}= L^{M-1} \sum_{l=1}^L R_l P_i\rho P_iR_l^{\dagger}$. Now we have
\begin{align}
\label{eq:diagonal-block}
    \sum_{i=1}^M \frac{1}{L^M}\sum_{k=1}^{L^M}V_k P_i \rho P_i V_k^{\dagger}
    &= \sum_{i=1}^M \frac{1}{L}\sum_{l=1}^{L} R_l P_i\rho P_i R_l^{\dagger} \nonumber \\
    &= \sum_i \int_{\wt{V}_i} \wt{V}_iP_i\rho P_i \wt{V}_i^{\dagger} ~\mathrm{d}\mu \left(\wt{V}_i\right),
\end{align}
where $\wt{V}_i \equiv \wt{V}_i\oplus_{j\neq i}0$.

\medskip
\noindent
{\bf {For the~$\mathbf{i\neq j}$} case}: For $i\neq j$, the action of $V_k=\left(\oplus_{i=1}^M \wt{V}_i\right)_k$ on $P_i\rho P_j$ corresponds to applying two unitaries: $R_l$ from left on block $i$ and $R_{l'}^{\dagger}$ from right on block $j$. Given any two choices of $R_l~\text{and}~ R_{l^{'}}$ in the $i^{th}$ and $j^{th}$ block, respectively, of $\oplus_{i=1}^M \wt{V}_i$ we can have $L^{M-2}$ possible choices of $V_l$ in the remaining blocks of $\oplus_{i=1}^M \wt{V}_i$. In particular, we have
\begin{align}
\label{eq:nondiagonal-block}
&\sum_{i\neq j} \frac{1}{L^M} \sum_{k=1}^{L^M}V_k P_i \rho P_j V_k^{\dagger}\nonumber\\
&=\sum_{i\neq j}\frac{1}{L^2} \sum_{l,l^{'}=1}^L R_l P_i\rho P_j R_{l^{'}}^{\dagger}\nonumber \\ 
&= \sum_{i\neq j} \int_{\wt{V}_i, \wt{V}_j} \wt{V}_iP_i\rho P_j \wt{V}_j^{\dagger} ~\mathrm{d}\mu \left(\wt{V}_i\right) ~\mathrm{d}\mu \left(\wt{V}_j\right).
\end{align}
Now combining Eqs. \eqref{eq:diagonal-block} and \eqref{eq:nondiagonal-block}, we get
\begin{align*}
& \frac{1}{L^M} \sum_{k=1}^{L^M}V_k \rho  V_k^{\dagger}=\int_V V \rho V^{\dagger} \prod_{i=1}^M\mathrm{d}\mu\left(\wt{V}_i\right),
\end{align*}  
where $V=\oplus_{i=1}^M\wt{V_i}$. This concludes the proof of the lemma.
\end{proof}

\setlength{\algomargin}{-0.5em}
\RestyleAlgo{boxruled}
\begin{protocol}[ht!!]
  \caption{$\texttt{Work Extraction using Single-Ancilla LCU}$}
  \label{algo:randomized-time-evolution}
   \begin{enumerate}
   \item With probability $1/2$, perform energy measurement on unknown state $\rho$.
   \item With probability $1/2$, execute steps $2(a)$ to $2(e)$.
    \begin{enumerate}
    \item Prepare the state $\rho_1=\ket{+}\bra{+}\otimes \rho$.
    \item Sample an integer $\ell$ uniformly at random from $[1, L^{M}]$ to select $Q_\ell=\{W_{j,\ell}\}_{j=1}^{2M}$.
    \item Obtain i.i.d.~samples $W_1, W_2$ from the distribution $\left\{W_{j,\ell}, 1/(2M)\right\}$.
    \item Define $\wt{W}_1=\ket{0}\bra{0}\otimes \mathbb{I}_N+\ket{1}\bra{1}\otimes W_1$ and $\wt{W}_2=\ket{0}\bra{0}\otimes W_2+\ket{1}\bra{1}\otimes \mathbb{I}_N$. Obtain the state $\rho'$ from $\rho$ as
    \begin{align*}
     \quad \quad \quad \quad \rho'=\left(\mathbb{I}_2\otimes \ug\right)\wt{W}_2 \wt{W}_1\rho_1 \wt{W}^{\dag}_1\wt{W}^{\dag}_2\left(\mathbb{I}_2\otimes \ugdag\right).
    \end{align*}
    Perform measurement of  $-X\otimes H$ on the state $\rho'$.
    \item Multiply $M^2$ to the measurement outcome.
  \end{enumerate}
  \item  For the $j^{\mathrm{th}}$ iteration, store the outcome $ \zeta_{j}$, which on an average (as Steps $1$ and $2$ occur  with probability $1/2$ each) takes value $\frac{1}{2}\left(\lambda_j+ M^2 \wt{\lambda}_j \right)$, where $\lambda_j\in\{E_1,\cdots,E_s\}$ and $\wt{\lambda}_j\in\{ \pm E_1,\cdots, \pm E_s\}$.
  \item Repeat Steps $1$ to $3$ for $T$ times to obtain values $\{\zeta_1,\cdots,\zeta_T\}$ and output $\overbar{\zeta} := \frac{2}{T}\sum_{j=1}^T \zeta_j$.
  
\end{enumerate}
\end{protocol}

\medskip
\noindent
{\it Integrated sampling strategy}:--We sample $\left\{\wt{V}_i\right\}_{i=1}^{M}$ from the Pauli group of $L$ elements. This set then naturally induces the set $\{W_i\}_{i=1}^{2M}$ of $2M$ elements (see Eq. \eqref{eq:w-with-2M}). There are $L^M$ such sets $\{W_i\}$, let us denote them by $Q_{\ell}=\{W_{i,\ell}\}_{i=1}^{2M}$ for $\ell\in\left[L^M\right]$. Thus, in Protocol \ref{algo:randomized-time-evolution}, we can effectively do the following: sample an integer $\ell\in \left[L^M\right]$ uniformly at random to select the set $Q_{\ell}$. Now, we can proceed by obtaining i.i.d. samples $W_1, W_2$ according to the ensemble $\left\{W_{i,\ell}, \frac{1}{2M}\right\}_{i=1}^{2M}$. Then by following the steps of Protocol \ref{algo:randomized-time-evolution}, we output some $\mu_j$ such that for any observable $O$,
\begin{align}
\mathbb{E}[\mu_j]=\Tr\left[O~\dfrac{1}{L^M}\sum_{\ell=1}^{L^M} \overbar{W}_{\ell}\rho \overbar{W}_{\ell}^{\dag}\right],
\end{align}
where $\overbar{W}_{\ell}=\sum_{j=1}^{2M}W_{j,\ell}/(2M)$. We can formally state the steps of Protocol  \ref{algo:randomized-time-evolution}.

\begin{remark}
In step 1 of Protocol \ref{algo:randomized-time-evolution}, the outcomes for the measurement of $H$ are the eigenvalues of $H$. In step 2(d), when we are measuring $-X \otimes H$, for each eigenvalue of $H$, there are two possible outcomes of $X$, namely, $\{+1, -1\}$.
\end{remark}

\begin{theorem}
\label{theorem:observe-erg}
Suppose $\varepsilon,\delta\in (0,1)$. Let $H$ be an $n$-qubit Hamiltonian and $\rho\in \mathcal{D}_N$ be an $N=2^n$-dimensional unknown density matrix. Then Protocol \ref{algo:randomized-time-evolution} outputs $\zeta$ satisfying,
$\left|\zeta-\ergobs{\rho}\right|\leq \varepsilon$ with success probability at least $1-\delta$ for 
\begin{align}
    T=O\left(\dfrac{M^4\norm{H}^2}{\varepsilon^2}\log\left(\frac{1}{\delta}\right)\right).
\end{align}
Further, if the unitary $\ug$ can be constructed with gate depth $\tau_g$ and $\tau_{\max}$ is the upper bound on the gate depth of the most expensive $W_{j,\ell}$, the cost of each run is at most $2\tau_{\max}+\tau_g$.
\end{theorem}

\begin{proof}
In each run of the protocol, with probability $1/2$, protocol measures $H=\sum_{a=1}^sE_a Q_a$ (see Eq. \eqref{eq:full-hamiltonian1}) in the state $\rho$. This results in an outcome $E_a$ with probability $q_a=\Tr\left[\rho Q_a\right]$, where $a\in[s]$. Step $2$, with probability $1/2$, performs measurement of $-X\otimes H$ on the state $\rho'$, where
\begin{align*}
    \rho'=\left(\mathbb{I}_2\otimes \ug\right)\wt{W}_2 \wt{W}_1\rho_1 \wt{W}^{\dag}_1\wt{W}^{\dag}_2\left(\mathbb{I}_2\otimes \ugdag\right).
\end{align*}
In the expression for $\rho'$, $\wt{W}_1$ and $\wt{W}_2$ are unitaries sampled with probability $1/(2M)$ each indirectly from $Q_l$, which itself is sampled with probability $1/L^M$. Simplifying $\rho'$, we have
\begin{align*}
    \rho'&:=\rho'(W_1, W_2; Q_{\ell})\\
    &=\left(\mathbb{I}_2\otimes \ug\right)\left[\ket{0}\bra{0}\otimes W_2 \rho W^{\dag}_2 +\ket{0}\bra{1}\otimes W_2 \rho W^{\dag}_1\right.\\
    &\left. +\ket{1}\bra{0}\otimes W_1 \rho W^{\dag}_2+\ket{1}\bra{1}\otimes W_1 \rho W^{\dag}_1\right]\left(\mathbb{I}_2\otimes \ugdag\right).
\end{align*}
We have $X=\sum_{r\in\{1,-1\}}r~ \Pi_r$, where $\Pi_1=\ket{+}\bra{+}$ and $\Pi_{-1}=\ket{-}\bra{-}$. Thus, measurement in step $2$ results in outcome $-M^2 r \times E_a$ with probability $\Tr\left[\rho' \Pi_r\otimes Q_a\right]:=q_{ra}$ for fixed $W_1, W_2$ and $Q_{\ell}$. Each run of steps $1$ and $2$ of the Protocol \ref{algo:randomized-time-evolution} results in an independent and identically distributed random variable. Let us call this random variable $Y$ and it takes values $E_a$ with  probability $q_a/2$ and $-M^2r~ E_a$ with probability $ q_{ra}/2$ for fixed $W_1, W_2$ and $Q_{\ell}$. That is, $\mathbb{E}\left[\zeta_j\right]=\mathbb{E}\left[Y\right]$ for all $j\in[T]$. This also implies that $\mathbb{E}\left[\overbar{\zeta}\right] = \frac{2}{T}\sum_{j=1}^T\mathbb{E}\left[\zeta_j\right] = 2\mathbb{E}\left[Y\right]$. The expectation value of  $Y$ is given by
\begin{align*}
    \mathbb{E}[Y]&=\frac{1}{2}\sum_{a=1}^s E_a \Tr\left[\rho Q_a\right]- \mathbb{E}_{W_1,W_2,Q_{\ell}} \frac{M^2}{2}\sum_{r\in\{\pm 1\}}\sum_{a=1}^s r E_a \Tr\left[\rho' \Pi_r\otimes Q_a\right]\\
    &=\frac{1}{2}\Tr\left[\rho H\right] - \frac{M^2}{2}\Tr\left[\mathbb{E}_{W_1,W_2,Q_{\ell}}\left[\rho'\right] X\otimes H \right]\\
    &=\frac{1}{2}\Tr\left[ H\left(\rho-\rho''\right)\right],
\end{align*}
where $\rho''=M^2\ug\mathbb{E}_{W_1,W_2,Q_{\ell}}\Tr'\left[\rho' X\otimes \mathbb{I}_N\right]\ugdag=M^2\ug\mathbb{E}_{Q_{\ell}}\left[\overbar{W}_{\ell} \rho \overbar{W}_{\ell}^{\dag}\right]\ugdag$, 
and $\Tr'$ denotes the trace with respect to first system. Further, $\overbar{W}_{\ell}=\sum_{j=1}^{2M}W_{j,\ell}/(2M)=\left(\oplus_{j=1}^M\wt{V}_j\right)_l/M = \left(V\right)_l/M$. Thus, we have
\begin{align*}
    \mathbb{E}\left[\overbar{\zeta}\right]&=2\mathbb{E}[Y]\\
    &=\Tr\left[ H \rho\right]-
    \frac{1}{L^M}\Tr\left[ H\ug \sum_{l=1}^{L^M}\left(V\right)_l \rho \left(V^{\dag}\right)_l \ugdag\right]\\
    &=\Tr\left[ H \left(\rho-\ug \wt{\rho} \ugdag\right)\right]\\
    &=\ergobs{\rho},
\end{align*}
where $V=\oplus_{j=1}^M\wt{V}_j$ and $\wt{\rho}$ is given by Eq. \eqref{eq:randomized-state}. But the protocol gives us the sample average, i.e, $\overbar{\zeta}=(2/T)\sum_{j=1}^T\zeta_j$ instead of $\mathbb{E}\left[\overbar{\zeta}\right]=\ergobs{\rho}$. The random variable $2Y$ takes values in the interval $\left[-2(1+M^2)\norm{H}, 2(1+M^2)\norm{H}\right]$. Now, applying Hoeffding's inequality to $\overbar{\zeta}$, we have
\begin{align*}
    \Pr\left[\left|\overbar{\zeta}-\ergobs{\rho}\right|\geq \varepsilon\right]\leq 2 \exp\left[\dfrac{- T\varepsilon^2}{8\left(1+M^2\right)^2\norm{H}^2}\right].
\end{align*}
Thus, taking
\begin{align*}
   T\geq  \dfrac{8\left(1+M^2\right)^2\norm{H}^2}{\varepsilon^2}\ln\left(\frac{2}{\delta}\right)
\end{align*}
ensures that 
\begin{align*}
    \Pr\left[\left|\overbar{\zeta}-\ergobs{\rho}\right|< \varepsilon\right]\geq 1-\delta.
\end{align*}
Thus, it is sufficient to choose
$$T=O\left(\dfrac{M^4\norm{H}^2}{\varepsilon^2}\log\left(\dfrac{1}{\delta}\right)\right)
$$ to estimate the observational egrotropy with an additive error $\varepsilon$ and success probability at least $1-\delta$ for all $\varepsilon>0$ and $0< \delta <1$.
\end{proof}
As mentioned previously, the sample complexity is exponentially worse than Protocol \ref{prot1} of the main article as the factor $M^4$ can scale as a polynomial in the dimension $N$.

\section{Robustness of the observational ergotropy: A simple example}
\label{append:ex-rob}
Consider a $4$-qubit system ($N=16$) with the coarse-grained state
\begin{align}
  \wt{\rho} = \sum_{i=1}^{4} \frac{p_i}{2^r}\, P_i,
\end{align}
where each block projector $P_i$ has $d=2^r=4$ dimensions and corresponds to a maximally mixed state in that subspace. Let the actual measurement probabilities be given as
\[
p_1 = 0.40, \quad p_2 = 0.37, \quad p_3 = 0.15, \quad p_4 = 0.08.
\]
Let the single qubit Hamiltonian be given by $\ket{1}\bra{1}$. Then the total $4$-qubit Hamiltonian has energy levels with degeneracies $g_0 = 1$, $g_1 = 4$, $g_2 = 6$, $g_3 = 4$ and  $g_4 = 1$ corresponding to energy eigenvalues $0$, $1$, $2$, $3$, and $4$, respectively. Let us define the following sets
\begin{align*}
    &A_1=\{\ket{0000}, \ket{0001}, \ket{0010}, \ket{0100}\},\\
    &A_2=\{\ket{1000}, \ket{0011}, \ket{0101}, \ket{0110}\},\\
    &A_3=\{\ket{1001}, \ket{1010}, \ket{1100}, \ket{0111}\},\\
    &A_4=\{\ket{1011}, \ket{1101}, \ket{1110}, \ket{1111}\}.
\end{align*}
The global  unitary $U_{\mathrm{gl}}$, estimating the observational ergotropy, maps eigenstates in block $P_i$ to eigenstates in $A_i$ for $i=1,2,3,4$. Then the observational ergotropy is given by
\begin{align}
    \ergobs{\rho} &= \Tr[\rho H] - \frac{1}{4}\left[0.4\times 3 + 0.37\times (1+2\times 3) + 0.15\times (2\times 3 + 3)+0.08\times (3\times 3 + 4)\right]\nonumber\\
    &= \Tr[\rho H] - \frac{1}{4}\left[1.20 + 2.59 + 1.35 +1.04\right]\nonumber\\
    &= \Tr[\rho H] - 1. 545.
\end{align}
Suppose the estimated probabilities have an error of order $0.05$, i.e. $\varepsilon=0.05$, and the estimated probabilities are obtained as
\[
\tilde{p}_1 = 0.36, \quad \tilde{p}_2 = 0.39, \quad \tilde{p}_3 = 0.18, \quad \tilde{p}_4 = 0.07.
\]
Here $\tilde{p}_2 > \tilde{p}_1$, so the changed global unitary $\wt{U}_{\mathrm{gl}}$, determining the guessed observational ergotropy, and the global unitary $U_{\mathrm{gl}}$ differ by a swap of blocks $P_1$ and $P_2$. $\wt{U}_{\mathrm{gl}}$ maps eigenstates in block $P_{\pi(i)}$ to eigenstates in $A_i$ for $i=1,2,3,4$. Here $\pi$ is a permutation of four objects defined as $\pi(1) = 2$, $\pi(2)=1$, $\pi(3)=3$ and $\pi(4)=4$. Then the guessed observational ergotropy is given by
\begin{align}
   \wt{\mathrm{Erg}}_{\mathrm{obs}}(\rho) &= \Tr[\rho H] - \frac{1}{4}\left[0.37\times 3 + 0.4\times (1+2\times 3) + 0.15\times (2\times 3 + 3)+0.08\times (3\times 3 + 4)\right]\nonumber\\
   &= \Tr[\rho H] - \frac{1}{4}\left[1.11+2.8+1.35+1.04\right]\nonumber\\
    &= \Tr[\rho H] - 1.575,
\end{align}
Thus, $\left|\wt{\mathrm{Erg}}_{\mathrm{obs}}(\rho)-\ergobs{\rho}\right|=0.03$. Now the upper bound in Eq. \eqref{eq:bound-good} of Lemma \ref{lem:guess-ergobs-ergobs-diff} is given by $\norm{H}\sum_{i=1}^4\left|p_i-p_{\pi(i)}\right| = 4\times2 |p_1-p_2|=0.24$.

\bibliography{ref}
\end{document}